\documentclass[11pt]{article}
\usepackage{graphicx}
\usepackage{amsfonts}
\usepackage{amsmath}
\usepackage{amssymb}
\usepackage{url}
\usepackage{fancyhdr}
\usepackage{indentfirst}
\usepackage{enumerate}
\usepackage{nameref}
\usepackage{mathtools}
\usepackage[margin=1in]{geometry}
\usepackage{dsfont}
\usepackage[colorlinks=true,citecolor=blue]{hyperref}
\usepackage{amsthm}
\usepackage{color}
\usepackage{natbib}
\numberwithin{equation}{section}

\usepackage{comment}
\usepackage{capt-of}
\def\d{\mathrm{d}}

\newcommand{\ES}{\mathrm{ES}}
\newcommand{\LES}{\mathrm{LES}}
\newcommand{\X}{\mathcal {X}}

\newcommand{\VaR}{\mathrm{VaR}}

\newcommand{\RVaR}{\mathrm{RVaR}}

\newcommand{\E}{\mathbb{E}}
\newcommand{\B}{\mathcal{B}}
\newcommand{\F}{\mathcal{F}}

\newcommand{\A}{\mathbb{A}}
\newcommand{\R}{\mathbb{R}}

\newcommand{\p}{\mathbb{P}}

\newcommand{\M}{\mathcal{M}}

\newcommand{\cU}{\mathcal{U}}

\newcommand{\id}{\mathds{1}}

\newcommand{\bsyb}{\boldsymbol}

\renewcommand{\(}{\left(}
\renewcommand{\)}{\right)}
\renewcommand{\[}{\left[}
\renewcommand{\]}{\right]}
\renewcommand{\ge}{\geqslant}
\renewcommand{\le}{\leqslant}
\renewcommand{\geq}{\geqslant}
\renewcommand{\leq}{\leqslant}
\renewcommand{\epsilon}{\varepsilon}
\newcommand{\esssup}{\mathrm{ess\mbox{-}sup}}
\newcommand{\essinf}{\mathrm{ess\mbox{-}inf}}
\renewcommand{\cdots}{\dots}

\theoremstyle{plain}
\newtheorem{theorem}{Theorem}
\newtheorem{corollary}{Corollary}
\newtheorem{lemma}{Lemma}
\newtheorem{proposition}{Proposition}
\theoremstyle{definition}
\newtheorem{definition}{Definition}

\newcommand{\Xc}{\mathcal{X}}
\newcommand{\Vc}{\mathcal{V}}

\newcommand{\Bc}{\mathcal{B}}

\newcommand{\Fc}{\mathcal{F}}

\newcommand{\Mc}{\mathcal{M}}

\newcommand{\Xf}{\mathbf{X}}

\newcommand{\Eb}{\mathbb{E}}
\newcommand{\Rb}{\mathbb{R}}
\newcommand{\Nb}{\mathbb{N}}
\newcommand{\Pb}{\mathbb{P}}

\newcommand{\Xv}{\vec{X}}

\newcommand{\tcrd}{\textcolor{red}}

\newtheorem{assumption}{Assumption}

\theoremstyle{remark}
\newtheorem{remark}{Remark}

\newcommand{\trd}{\textcolor{red}}
\newcommand{\cet}{\begin{center}}
\newcommand{\ecet}{\end{center}}

\definecolor{olive}{rgb}{0.3, 0.4, .1}
\definecolor{fore}{RGB}{249,242,215}
\definecolor{back}{RGB}{51,51,51}
\definecolor{title}{RGB}{255,0,90}
\definecolor{dgreen}{rgb}{0.,0.6,0.}
\definecolor{gold}{rgb}{1.,0.84,0.}
\definecolor{JungleGreen}{cmyk}{0.99,0,0.52,0}
\definecolor{BlueGreen}{cmyk}{0.85,0,0.33,0}
\definecolor{RawSienna}{cmyk}{0,0.72,1,0.45}
\definecolor{Magenta}{cmyk}{0,1,0,0}

\newcommand{\dsquare}{\mathop{  \square} \displaylimits}

\usepackage{setspace}
\newcommand{\com}[1]{\marginpar{{\begin{minipage}{0.18\textwidth}{\setstretch{1.1} \begin{flushleft} \footnotesize \color{red}{#1} \end{flushleft} }\end{minipage}}}}

\begin{document}
	
\title{Extended Convolution Bounds on the Fr\'echet Problem: \\Robust Risk Aggregation and %Pareto-Optimal
Risk Sharing} 

\author{
Peng Liu\thanks{\scriptsize School of Mathematics, Statistics and Actuarial Science, University of Essex, UK. Email: \texttt{peng.liu@essex.ac.uk}}
\and
Yang Liu\thanks{\scriptsize School of Science and Engineering, The Chinese University of Hong Kong (Shenzhen), China. Email: \texttt{yangliu16@cuhk.edu.cn}}
\and 
Houhan Teng\thanks{\scriptsize School of Science and Engineering, The Chinese University of Hong Kong (Shenzhen), China. Email: \texttt{s2447087@link.cuhk.edu.cn}}
}

\date{}%\today}

\maketitle
\begin{abstract}
    In this paper, we provide extended convolution bounds for the Fr\'{e}chet problem and discuss related implications in quantitative risk management. First, we establish a new inequality for the Range-Value-at-Risk ($\RVaR$). Based on this inequality, we obtain 
    bounds for robust risk aggregation with dependence uncertainty for (i) $\RVaR$, (ii) inter-RVaR difference %the difference of two $\RVaR$
    and (iii) inter-quantile difference, and provide sharpness conditions. 
    %the difference of two quantiles,  
    These bounds are called extended convolution bounds, which not only complement the results in the literature (convolution bounds in \citet{BLLW24}) but also offer results for some variability measures. Next, applying the above inequality, we study the risk sharing problem for averaged quantiles (corresponding to risk sharing for distortion risk measures with special inverse S-shaped distortion functions), which is a non-convex optimization problem. We obtain an expression of the minimal value of the risk sharing and an explicit expression for the corresponding optimal allocation, which is comonotonic risk sharing for large losses and counter-monotonic risk sharing for small losses or large gains. Finally, we explore the dependence structure for the optimal allocations, showing that the optimal allocation does not exist if the risk is not bounded from above.
    
    \begin{bfseries}Keywords\end{bfseries}: Robust risk aggregation; Risk sharing; 
    Range-Value-at-Risk ($\RVaR$); %Value-at-Risk;
    Quantiles; Inter-$\RVaR$ difference; Inter-quantile difference; Distortion risk measures; Dependence uncertainty %Inf-convolution
\end{abstract}

\section{Introduction}\label{sec:intro}

The \cite{F51} problem in probability theory concerns the characterization of the possible distributions of $f(X_1, \dots, X_n)$ when the distributions of the random variables $X_1, X_2, \dots, X_n$ are known, but their joint dependence structure is unspecified, where $f: \R^n \to \R$ is a measurable  function. Typically, $f$ is the sum, and we denote the sum variable by $S = X_1 + X_2 + \dots + X_n$. Formally, given the set of feasible joint distributions with some fixed marginal distributions (formulated from \cite{D56}), the problem seeks to determine sharp bounds on functionals of $S$, such as its cumulative distribution function (cdf) or quantiles. This problem dates back to \cite{H40} and the classical Fr\'echet-Hoeffding bounds establish the preliminary theory on the possible dependence structure. Further, if $n=2$, analytical results on the largest possible distribution (equivalently, the largest quantiles) of $S$ were derived in \cite{M81} and \cite{R82}. However, if $n \geq 3$, the problem becomes substantially more complex when optimizing among all admissible copulas to obtain the best-possible bounds for some specific functionals. This problem has a deep connection with measure theory, functional analysis and optimal transport, as it involves constrained optimization over spaces of probability measures; see \cite{EP06} and \cite{NW22}. 

Beyond its theoretical significance, the Fr\'echet problem of the sum variable plays a crucial role in risk management, where dependence assumptions significantly impact the evaluation of risk measures on $S$ such as Value-at-Risk (VaR; quantile) and Expected Shortfall (ES). In practice, data from different correlated products are often collected separately and thus no dependence information is available; see \cite{EPR13} and \cite{EWW15}. Even if the data are available, the estimation of the dependence structure typically has low accuracy, resulting in substantial uncertainty on the choice of the dependence structure; see e.g., Chapter 8 of \cite{MFE15}. Consequently, sharp bounds for some risk functionals under dependence uncertainty indicate the worst/best-case (robust) risk aggregation and evaluation, which is hence of great significance. Besides, 
the Fr\'echet problem has a wide application in different areas of operations research, including assembling line crew scheduling (\cite{H84}), matching theory (\cite{BTWZ23}), worst-case portfolio selection (\cite{CLLW22}), multiple statistical hypothesis testing (\cite{VWW22}), etc.; a comprehensive discussion was given in \cite{BLLW24}. 
Consequently, recent technical advances in optimal transport (\cite{BKLP22}), copula theory (\cite{KLW24}) and numerical optimization (\cite{SDR21}) have been incorporated into the study of the Fr\'echet problem, thereby enabling more precise characterizations in probability theory and furnishing additional tools for these various disciplines. 

In the field of quantitative risk management, the Fr\'echet problem is often referred to as the risk aggregation problem with dependence uncertainty, which is inherently challenging due to its nature. 
Explicit expressions are only available in some special cases of the marginal distributions. For general marginal distributions, only bounds are available for robust quantiles in the literature; see e.g., the dual bounds in Theorem 4.17 of \cite{R13}. The explicit expression for robust quantiles is only available for marginal distributions with monotone densities; see e.g., \cite{WPY13}, \cite{BJW14} and \cite{JHW16}, where the worst-case dependence structure is a combination of joint-mixability (see \cite{WW16}) and mutual exclusivity (introduced in \cite{DD99}). Recently, \cite{BLLW24} offered a so-called convolution bound, which is proved to be the sharp bound for robust $\RVaR$ and quantiles under some specific cases, especially including the case that all the marginal distributions have monotonic densities in the same direction on their tail parts. Here, $\RVaR$ (Range-Value-at-Risk) is a two-parameter class of non-convex risk measures, including both $\VaR$ and $\ES$ as special cases, which was first introduced by \cite{CDS10} and will be defined in \eqref{eq:r1}. In general, computational and optimization approaches, such as rearrangement algorithms (\cite{EPR13}), scheduling (\cite{BJV18}), neural networks (\cite{EKP20}), and linear programming formulations (\cite{AB21}), have been developed to approximate the sharp bounds numerically, albeit with their own drawbacks. It is worth noting that the convolution bound is closely linked to the $\RVaR$ inequality in \cite{ELW18}, which is used to address risk-sharing problems among multiple agents with risk preferences characterized by $\RVaR$.  
%\cite{BLLW24} further extended the above inequality to the convolution bound, where the sufficient condition on the dependence structure of the aggregated risk for its sharpness is also studied. 

%As illustrated later in the paper, 
Indeed, the Fr\'echet problem offers a distinctive perspective on studying the risk sharing problem, as risk sharing can be viewed as the  ``inverse” of risk aggregation. The risk sharing problem concerns redistributing a total risk among multiple participants, requiring the determination of the optimal allocations. This problem has a long history, dating back to the seminal work of \cite{B62}, which studied risk sharing through the framework of expected utility. 
%\cite{ADEH99}, \cite{FS02} and \cite{FR05} introduced the coherent or convex risk measures and for coherent or convex risk measures or risk measures satisfying convexity. 
%Note that the optimal allocation is \emph{comonotonic} with the law-invariant assumption for these risk measures. 
Over the past two decades, researchers have explored risk sharing problems using risk measures to represent participants' risk preferences since the introduction of coherent and convex risk measures in \cite{ADEH99}, \cite{FS02} and \cite{FR05}. For instance,  \cite{BE05}, \cite{JST08}, \cite{FS08} and \cite{D12} examined risk sharing problems based on risk measures satisfying convexity and law invariance, leading to {comonotonic} optimal allocations. Moreover, \cite{ELW18} investigated risk sharing for $\RVaR$, derived the minimal value of risk sharing (called inf-convolution) and provided explicit expressions for optimal allocations in both cooperative and competitive settings; see also \cite{ELMW20}. Besides, some other non-convex risk measures were also studied in the risk sharing problem such as $\VaR$-type distortion risk measures (\cite{W18}),  Lambda $\VaR$ (\cite{L25}), inter-quantile difference (\cite{GLW23}) and general non-convex risk measures (\cite{L22}).
%Risk sharing problems with $\RVaR$ for both cooperative and competitive agents were studied to \cite{ELW18} found the minimal value of risk sharing for cooperative agents and the explicit expressions of optimal allocations for both cooperative and competitive agents. 

%Interestingly, it is shown in  \cite{ELW18} and \cite{BLLW24} that both risk sharing and risk aggregation rely on an $\RVaR$ inequality, which works as the starting point in those two papers. 

In this paper, we focus on the theoretical aspect of the Fr\'{e}chet problem rather than applications, and we elaborate on our contribution as follows. 
% dependence structure of the risks for  $\RVaR$ inequality being an equality is fully discussed in Section \ref{Sec:DS}. 
%The cornerstone of our work is a new extension of the RVaR inequality, which could be regarded as a symmetric version of the convolution bound.
First, we establish a new type of upper and lower bounds for robust risk aggregation with $\RVaR$ under dependence uncertainty in Section \ref{sec:convolution1}. We show that those bounds are sharp as upper bounds if the marginal distributions have increasing densities on their upper-tail parts and also sharp as lower bounds if the marginal distributions have decreasing densities on their lower-tail parts. These bounds use a general structure of averaged quantiles to obtain new forms of bounds, and offer new sharpness results for robust RVaR aggregation compared with those in \cite{BLLW24}.  Specifically, while the convolution bound for robust RVaR aggregation in \cite{BLLW24} is a sharp lower bound if the marginal distributions exhibit increasing densities on their lower-tail parts, our derived bound achieves sharpness as a lower bound if marginal distributions have decreasing densities on their lower-tail parts.

Second, %we find that the worst-case of the difference of two $\RVaR$ under dependence uncertainty equals the difference of the worst-case of one $\RVaR$ and the best-case of the second $\RVaR$. This conclusion also holds true for the difference of two quantiles. Based on this fact, combining our new upper and lower bounds for $\RVaR$ with the convolution bounds in \cite{BLLW24}, we obtain the   the sharp bounds for robust risk aggregation of the difference of two $\RVaR$ and the difference of two quantiles. 
%In Section \ref{Sec:inequality}, we firstly obtain a new $\RVaR$ inequality.  
%whose extremal dependence structure (also studied in the same section) is ``symmetric" to the ones in Convolution Bound.  
 %The first theorem of this paper, Theorem \ref{thm: newRVaRinequality}, is an extension of the convolution bound. 
we obtain two more types of upper bounds for the difference between two $\RVaR$s and that between two quantiles, called inter-$\RVaR$ difference ($\mathrm{IRD}$) and, for a special case, inter-quantile difference ($\mathrm{IQD}$)  respectively %for a special case
in Section \ref{Sec:Robust}. Note that $\mathrm{IRD}$ extends the inter-$\ES$  proposed in \cite{BFWW22} as an alternative to the standard deviation to measure the variability of risks. It is well known that $\mathrm{IQD}$ is frequently used to find the outliers and measure the statistical dispersion in statistics. Moreover, $\mathrm{IQD}$ is also a tool to quantify the variability of risks; see \cite{BFWW22} and \cite{GLW23}.  Importantly, we find that the robust $\mathrm{IRD}$ (resp., $\mathrm{IQD}$) equals the difference between the two robust $\RVaR$ (resp., quantiles). Hence, the explicit expressions are derived based on the robust $\RVaR$ and quantiles for marginal distributions with monotone densities on both upper- and lower-tail parts. We offer two different sharp upper bounds for $\mathrm{IRD}$ under different assumptions on the marginal distributions: the first condition requires the marginals to have decreasing densities on their upper-tail parts and increasing densities on their lower-tail parts; whereas the second requires decreasing densities on both upper- and lower-tail parts, respectively. Commonly used continuous distributions in finance or risk management (e.g., Gaussian, lognormal, t, exponential, and Pareto) %as special cases
mostly fall into either of the two categories. 
Those two expressions for the sharp upper bound of $\mathrm{IRD}$  together with their different assumptions demonstrate the complexity of the robust risk aggregation for $\mathrm{IRD}$ and also the usefulness of our new bounds established in Section \ref{sec:convolution1}. The sharp upper bound for the difference between two quantiles requires the marginals to have densities that are monotone in the same direction on both upper- and lower-tail parts, respectively, which is valid for almost all the commonly used continuous distributions. All three types of bounds established in our paper are generally called extended convolution bounds, as they can be viewed as an extension of the convolution bounds from $\RVaR$ and quantiles to the corresponding variability measures introduced in \cite{BFWW22}.

Third, 
%In Section \ref{sec:risksharing1}, we study the risk sharing for cooperative agents, where the preference of each agent is represented by  $R_{I_i}$. This problem is also called the inf-convolution in the literature; see e.g., \cite{D12} and \cite{R13}.  
we study the risk sharing problem for some averaged quantiles, which is equivalent to the risk sharing problem for distortion risk measures with some special inverse S-shaped distortion functions. This class of distortion risk measures represents the decision maker's typical attitude: risk aversion for large losses and risk-seeking for small losses or gains; see \cite{Y87} and \cite{TK92}. Clearly, this problem is non-convex and challenging.  It turns out that the inf-convolution %of $R_{I_i}$ 
has a very simple form: the lower-tail $\ES$, which is an averaged quantile below a specified quantile of the total risk.  Moreover, the optimal allocation, which is Pareto-optimal, exists if and only if the total risk is bounded from above. The structure of the optimal allocation consists of two parts: The upper-tail part of the total risk is shared comonotonically, and the lower-tail part of the risk is shared counter-monotonically. This optimal allocation is consistent with the agents' risk attitudes:  risk-aversion for large losses and risk-seeking for small losses or large gains. If the total risk is not bounded from above, the optimal allocation does not exist, where the proof is based on the analysis of the dependence structure of the optimal allocations.  Instead, in this case, we find a sequence of sub-optimal allocations such that the risk exposures generated by those allocations converge to the inf-convolution. We emphasize that our optimal allocation is a new combination of existing optimal allocations for convex risk measures and quantile-based risk measures in the literature; see e.g., \cite{JST08} and \cite{ELW18}. Although we only solve the risk sharing problem for distortion risk measures with some special inverse S-shaped distortion functions, we emphasize that to the best of our knowledge, this is the first non-constrained risk sharing result for this class of distortion risk measures with the distortion functions exaggerating the probability of large losses and the probability of large gains simultaneously. It sheds light on the further investigation into the risk sharing problem for general inverse-S-shaped distortion risk measures.

The rest of the paper is organized as follows. We give some notation and definitions in Section \ref{sec:notation}. We establish a new inequality for $\RVaR$ in Section \ref{Sec:inequality}. Based on this new inequality, we obtain some bounds on risk aggregation for $\RVaR$, the difference between two $\RVaR$s and the difference between two quantiles in Section \ref{sec:aggregation}. Employing this inequality, we study the risk sharing problem for the averaged quantiles and obtain the condition for the existence of the Pareto-optimal allocations and the forms of the Pareto-optimal allocations in Section \ref{sec:RiskSharing}.  Section \ref{sec:conc} concludes the paper.

\section{Notation and Definitions}
\label{sec:notation}

Let $(\Omega, \Fc, \Pb)$ be an atomless probability space and  $\Xc = L^0(\Omega, \F, \Pb)$ be the set of all random variables, and $\X_1 = L^1(\Omega, \F, \p)$ be the set of all random variables with finite mean. Correspondingly, we denote the set of all distributions of the random variables in $\Xc$ by $\Mc$, and by $\M_1$ the set of all distributions of the random variables in $\X_1$. To ease the notation, we treat almost surely equal random variables and events as identical and set $[n]: = \{1, \dots, n\}$ for $n \in \Nb^+.$  Throughout this paper, we use $\mathrm{U}[a,b]$ with $a<b$ to represent the uniform distribution on $[a,b]$. For any $X\in \Xc$, let $U_X \sim \mathrm{U}[0,1]$ such that $F_X^{-1}(U_X)=X$ and $U_X$ is usually called the probability transform of $X$ (e.g., Proposition 7.2 of \cite{MFE15}); The existence of such $U_X$ is guaranteed (e.g.,
Lemma A.32 of \cite{FS16}). In this paper, we identify a probability measure $\mu$ with its distribution function $F$ when no confusion arises. %the probability measure ($\mu$) and the distribution function ($F$) are considered equivalent. 
Either of them may be used according to the context. Moreover, terms like ``increasing'' and ``decreasing'' are in the non-strict sense.

%For $\boldsymbol \mu= (\mu_1,\dots,\mu_n)\in \M^n$, let $\Gamma (\boldsymbol \mu)$ be the set of probability measures on $\R^n$ that have one-dimensional marginals $\mu_1,\dots,\mu_n$
Next, we introduce a family of risk measures: the average quantile functional $R$. For any $I \in \Bc([0,1])$ with $|I|>0$, the functional $R_I: \Mc \rightarrow \Rb \cup \{\pm \infty \}$ is defined as
    \begin{equation*} 
        \label{def: R}
        R_{I} (\mu) = \frac{1}{|I|} \int _{I} q^-_t(\mu)\d  t,
    \end{equation*}
where $|I|$ is the length of $I$ under the Lebesgue measure on $\R$ and 
$q^-_t(\mu)=\inf\{x\in \R: \mu((-\infty,x])\ge t\}$ is the left quantile of distribution $\mu$ with $t \in (0, 1]$. As $R_I$ is a law-invariant risk measure, we abuse the notation $R_I(X)$ with $R_I(\mu)$ for a random variable $X\sim \mu$. This abuse may also apply to other law-invariant risk measures in our paper.  

%\tcbg{(Remark from TENG: the latestly discussed generalisation of the definition of RVaR is in fact the same as the original one. On the other hand, I believe by definition the $R$ does not integrate end-points and I do not understand how to integrate via Lebesgue integration.)} 

\begin{remark}
    \begin{enumerate}[(i)]
        \item In fact, one could also define $R_I$ using the right quantile $q_t^+(\mu)$ (i.e., $q_t^+(\mu) = \inf\{x \in \mathbb{R} : \mu((-\infty, x]) > t\}$), which does not affect the value of the functional as $q_t^-(\mu)\neq q_t^+(\mu)$ only holds at a countable number of points over $(0, 1)$. 
        \item Further, let $I \subseteq [0 ,1]$ be a union of finitely many intervals and denote by $\bar{I}$ the closure of $I$. As the integral value does not change if the integral region is changed within countably many points, we have $R_I(\mu) = R_{\bar{I}}(\mu)$. 
        Thus, without loss of generality, we always write any interval in $I$ %the functional $R$ 
        as a closed one. 
    \end{enumerate}
\end{remark}
The average quantile functional includes the classic risk measure, Range-Value-at-Risk ($\RVaR$), as a special case in the following way:
\begin{equation}\label{eq:r1} 
     \RVaR_{\beta,\beta+\alpha}(\mu)=\frac{1}{\alpha} \int _{\beta} ^{\beta+\alpha} q^+_{1-t}(\mu)\d t=R_{[1-\beta-\alpha, 1-\beta]} (\mu) , ~~\mu\in\M,
\end{equation} 
where $ 0 < \beta < \beta+\alpha \le 1$. Note that $\RVaR$ was first introduced by \cite{CDS10} as a family of robust risk measures, and it was further applied to the risk sharing and optimal reinsurance problem as the preference functional in \cite{ELW18} and \cite{FHLX25}. Note that $\RVaR$ includes $\mathrm{ES}$ and $\mathrm{LES}$ (Left ES) as special cases in the following way:
For $p\in (0,1]$,  
$$
\ES_p(\mu) = \frac{1}{p}\int_{1-p}^{1} q_u^- (\mu) \d u = R_{[1-p,1]}(\mu) \mbox{~~~and~~~} \LES_p(\mu) = \frac{1}{p}\int_{0}^{p} q_u^- (\mu) \d u = R_{[0,p]}(\mu), ~~  \mu \in \M_1.
$$
Moreover, both $q_\beta^-(\mu)$ and $q_\beta^+(\mu)$ for $\beta \in(0,1)$ appear as the limits of some $R_I$  via
\begin{equation*}
\begin{aligned}
	&%\lim_{ \alpha\downarrow 0} R_{1-\beta,\alpha}(\mu)= 
    \lim_{ \alpha\downarrow 0} R_{[\beta-\alpha, \beta]} (\mu) = q_{\beta}^-(\mu)
	\mbox{~~~and~~~}
	%\lim_{ \alpha\downarrow 0} R_{1-\beta-\alpha,\alpha}(\mu) = 
    \lim_{ \alpha\downarrow 0} R_{[\beta, \beta+\alpha]}(\mu) = q_{\beta}^+(\mu), ~~ \mu \in \M.%\label{eq:quantileconvergence}
\end{aligned}
\end{equation*}
%{\color{red}But note that if $I$ is a singleton, then $R_I(\mu) = 0$ for any $\mu \in \M$.}  
%One could easily see that all the aforementioned risk measures are law-invariant.

For $\boldsymbol \mu= (\mu_1,\dots,\mu_n) \in \M^n$, let $\Gamma (\boldsymbol \mu)$ be the set of probability measures on $\R^n$ with one-dimensional marginals $\mu_1,\dots,\mu_n$. For a probability measure $\mu $ on $(\R^n, \B(\R^n))$, define   $\lambda_{\mu}\in \M$ via $$\lambda_\mu ((-\infty,x])  = \mu(\{ (x_1,\dots,x_n)\in \R^n: x_1+\dots+x_n\le x\}),~ x\in \R.$$
In other words, $\lambda_\mu$ is an aggregated probability measure of the sum variable $\sum_{i=1}^n X_i$, where the random vector $(X_1,\dots,X_n)$ follows the $n$-dimensional distribution $\mu$.
Moreover, let $\Lambda (\boldsymbol \mu) =\{\lambda_\mu: \mu\in \Gamma (\boldsymbol \mu)\}$ be the set of all aggregated probability measures with specified marginals $\boldsymbol \mu$.
 Define an approximate standard simplex 
$$
\Delta_n= \left\{(\beta_0,\beta_1,\dots,\beta_{n}) \in (0,1)\times [0,1)^{n}: \sum_{i=0}^{n} \beta_i =1 \right\}.$$
Note that $\Delta_n$ is neither open nor closed; hence, we use the term ``approximate". 
For real numbers $x_i, i \in [n]$, we use the notation $\bigvee_{i=1}^n x_i=\max_{i\in [n]} x_i$ and $\bigwedge_{i=1}^n x_i=\min_{i \in [n]} x_i$. 
\begin{comment}
Finally, for $r\in (0,1)$ and a distribution $F$, let 
\begin{equation}
\label{eqn: truncated distribution function}
    F^{r+}(x)=\frac{(F(x)-r)_+}{1-r},~x\in\R,
\end{equation}
which is called the $r$-tail distribution of $F$  by \cite{RU02}. In other words, $F^{r+}$ is the distribution of the random variable $q^-_{U}(F)$ where $U$ is a uniform random variable on $[r,1]$.
Equivalently, $F^{r+}$ is the distribution of $F$ restricted beyond its $r$-quantile. %, while $\mu^{t-}$ is the distribution measure of $\mu$ restricted below its $t$-quantile.
Particularly, a statement that $F$ admits a decreasing density beyond its $r$-quantile is equivalent to the one that $F^{r+}$ admits a decreasing density on the whole support. 
\end{comment}
Finally, for $\mu \in \M$ and $r\in (0,1)$, let $\mu^{r+}$ be the probability measure given by
\begin{equation}
\label{eqn: truncated distribution function}
\mu^{r+} (-\infty,x] =\max\left\{ \frac{\mu(-\infty,x]-r}{1-r}, 0\right\},~~x\in \R,
\end{equation}
which is called the $r$-tail distribution of $\mu$ in \cite{RU02}. Indeed, $\mu^{r+}$ is the distribution measure of the random variable $q^-_U(\mu)$, where $U\sim \mathrm{U}[r,1]$. Equivalently, $\mu^{r+}$ is the distribution measure of $\mu$ restricted beyond its $r$-quantile (assuming no mass at this point). In this paper, a statement that $\mu$ admits a decreasing (resp., increasing) density beyond its $r$-quantile is equivalent to the one that $\mu^{r+}$ admits a decreasing (resp., increasing) density on its support. An analogous definition applies to the probability measure $\mu^{r-}$: %Let $\mu^{t-}$ be the probability measure given by
$$
\mu^{r-} \left(-\infty, x\right]= \min\left\{ \frac{\mu\left(-\infty, x\right]}{r}, 1\right\},~~x\in \R.
$$
That is, $\mu^{r-}$ is the distribution measure of the random variable $q^-_V(\mu)$, where $V \sim \mathrm{U}[0,r]$. A statement that $\mu$ admits a decreasing (resp., increasing) density below its $r$-quantile is equivalent to the one that $\mu^{r-}$ admits a decreasing (resp., increasing) density on its support. We use the endpoint conventions
$$
q_0^+(\mu)=\lim_{t\downarrow0}q_t^+(\mu),\qquad
q_1^-(\mu)=\lim_{t\uparrow1}q_t^-(\mu),
$$
which may take the values $-\infty$ and $+\infty$, respectively.
Moreover, we set
$
\mu^{0+}:=\mu,\; \mu^{1-}:=\mu.
$

% for $0 \le \alpha < \beta \le 1$
% and any  probability  measure $\mu$, we  let $\mu^{[\alpha, \beta]}$ be the probability measure given by
% $$
% \mu^{[\alpha, \beta]}(-\infty,x]=\frac{\left(\min\left\{\mu(-\infty,x], \beta\right\} -\alpha\right)_+}{\beta-\alpha},~~x\in \R.
% $$
% Equivalently, $\mu^{[\alpha, \beta]}$ is the distribution measure of the random variable $q_V (\mu)$ where $V \sim \mathrm \mathrm{U}[\alpha, \beta]$, a uniform random variable on $[\alpha, \beta]$. In particular, $\mu^{[\alpha, 1]} = \mu^{\alpha+}$ is the $\alpha$-tail distribution   of $\mu$ in Section \ref{sec:rvar}.

%For a set $I\subset [0,1]$ and a random variable $X$, define the risk functional$$R_{I} (X) = \frac{1}{|I|} \int _{I} q^-_t(X)\d  t,$$ where $|I|>0$ is the length of $I$ under Lebesgue measure and $q^-_t(X)=\inf\{x\in \R: F_X(x)\ge t\}$ is the left quantile of $F_X$, which is the distribution function of $X$, with $t \in (0, 1]$.

% =============================================================================
% Quantile Inequalities
% =============================================================================
\section{%Quantile inequalities
New RVaR Inequality}
\label{Sec:inequality} %\marginpar{\color{red}Extended or new?}
In this section, we establish a new inequality for $\RVaR$. This inequality will play a crucial role in establishing the bounds for risk aggregation with dependence uncertainty and in analyzing the risk sharing problem later.
%{\color{red} The case $\beta_i=0$ is not defined.} 
Before showing our result, we first display Propositions 1 and A.1 of \cite{BLLW24} (originally from Theorem 4.1 of \cite{LW21}) as below, which will be used frequently later. 
\begin{lemma}[Propositions 1 and A.1 of \cite{BLLW24}]\label{Le:transform} Suppose $\boldsymbol \mu=(\mu_1,\dots,\mu_n)\in \M^n$ and $0\leq r<r+s\leq 1$. Define $\boldsymbol \mu^{r+} = (\mu_1^{r+}, \dots, \mu_n^{r+})$ and $\boldsymbol \mu^{(r+s)-} = (\mu_1^{(r+s)-}, \dots, \mu_n^{(r+s)-})$. We have 
	$$
	\sup_{\nu\in \Lambda(\boldsymbol \mu) }R_{[r,r+s]} (\nu) = \sup_{\nu\in \Lambda(\boldsymbol \mu^{r+ } )}\LES_{\frac{s}{1-r}} (\nu),~
	\inf_{\nu\in \Lambda(\boldsymbol \mu) } R_{[r,r+s]} (\nu) = \inf_{\nu\in \Lambda(\boldsymbol \mu^{(r+s)-}) } \ES_{\frac{s}{r+s}} (\nu).
	$$
 
%	and  
%	$$	\inf_{\nu\in \Lambda(\boldsymbol{\mu})} q_t^- (\nu) = \inf_{\nu\in \Lambda(\boldsymbol \mu^{t- } )}q_1^- (\nu).$$ 
\end{lemma}

Next, we offer a new inequality for $\RVaR$. %Hereafter, we use the convention that $\infty-\infty=\infty$.
\begin{theorem}[New $\RVaR$ Inequality]
    \label{thm: newRVaRinequality}
 Let  $0<r < r + s \leq 1$,  $\alpha_1, \cdots, \alpha_n\geq 0$ and $\beta_1,\cdots, \beta_n>0$ satisfying $\sum_{i=1}^n \alpha_i + \bigvee_{i=1}^n \beta_i \leq 1-r$ and $\sum_{i=1}^n\alpha_i  \leq s$. Then for $\boldsymbol{\mu} \in \M^n$ and $\nu\in \Lambda (\boldsymbol \mu)$, we have 
    \begin{equation}
        \label{eqn: GeneralisedNewRVaR}
        R_{[r, r + s]}(\nu) 
        \leq 
        \sum_{i=1}^n \left[ 
        \frac{1-r-\beta_i}{s} R_{[r, r+\alpha_i] \cup [r+ \alpha_i+\beta_i, 1]} \left(\mu_i \right)
        +
        \left(1-\frac{1-r-\beta_i}{s}\right) R_{[r + \alpha_i, r+ \alpha_i+\beta_i]} \left(\mu_i\right) 
        \right].
    \end{equation}
    Moreover, \eqref{eqn: GeneralisedNewRVaR} holds for $r=0$ if $\boldsymbol{\mu} \in \M_1^n$.
\end{theorem}

\begin{proof}  
First, we suppose that $\boldsymbol{\mu}\in \mathcal M_1^n$ and
$r>0$. We begin by showing an auxiliary inequality \eqref{eq:boundary-ineq}. Let $Y_i\in L^1$, $i\in[n]$,
$Y=\sum_{i=1}^n Y_i$, and let $a_i\ge0$, $b_i>0$, $C\in(0,1]$. Put
$
A=\sum_{i=1}^n a_i$ and $b_*=\bigvee_{i=1}^n b_i.
$
We first assume the strict conditions
$$
A<C<1,\qquad A+b_*<1.
$$
Let $t=C-A$. Then $0<t<1-A$ and $b_*\vee t<1-A$. By Theorem 1 of
\cite{ELW18}, for $b_*\le u\le 1-A$, we have
$$
R_{[1-A-u,1-A]}(Y)
\le
\sum_{i=1}^n R_{[1-a_i-b_i,1-a_i]}(Y_i).
$$
Moreover, for $b_*\vee t<u\le 1-A$,
$$
R_{[1-A-u,1-A-t]}(Y)
\le
R_{[1-A-u,1-A]}(Y)
\le
\sum_{i=1}^n R_{[1-a_i-b_i,1-a_i]}(Y_i).
$$
Replacing $Y_i$ by $-Y_i$, $i\in[n]$, yields
$
R_{[A+t,A+u]}(Y)
\ge
\sum_{i=1}^n R_{[a_i,a_i+b_i]}(Y_i).
$ 
Letting $u=1-A$, we get
$
R_{[C,1]}(Y)
\ge
\sum_{i=1}^n R_{[a_i,a_i+b_i]}(Y_i).
$ 
Using
$$
\mathbb E(Y)
=
(1-C)R_{[C,1]}(Y)+C R_{[0,C]}(Y),
$$
we obtain
$$
C R_{[0,C]}(Y)
\le
\mathbb E(Y)
-
(1-C)\sum_{i=1}^n R_{[a_i,a_i+b_i]}(Y_i).
$$
Since
$$
\mathbb E(Y_i)
=
(1-b_i)R_{[0,a_i]\cup[a_i+b_i,1]}(Y_i)
+
b_iR_{[a_i,a_i+b_i]}(Y_i),
$$
it follows that
\begin{equation}
\label{eq:boundary-ineq}
R_{[0,C]}(Y)
\le
\sum_{i=1}^n
\left[
\frac{1-b_i}{C}
R_{[0,a_i]\cup[a_i+b_i,1]}(Y_i)
+
\left(1-\frac{1-b_i}{C}\right)
R_{[a_i,a_i+b_i]}(Y_i)
\right].
\end{equation}

We next extend \eqref{eq:boundary-ineq} to the boundary cases. We claim
that \eqref{eq:boundary-ineq} remains valid whenever
$$
A\le C\le1,\qquad A+b_*\le1.
$$
Indeed, for $Z\in L^1$, the function
$
H_Z(u):=\int_0^u q_v^-(Z)\,\mathrm d v,\; u\in[0,1],
$
is continuous. Therefore all terms in \eqref{eq:boundary-ineq}, written
in terms of $H_{Y_i}$, are continuous in $C,a_i,b_i$ as long as $b_i>0$,
with the usual interpretation of zero-length intervals through the
corresponding integral expressions. Choosing sequences
$
a_i^{(m)}\to a_i,\; b_i^{(m)}\to b_i,\; C_m\to C
$
such that
$
\sum_{i=1}^n a_i^{(m)}<C_m<1$ and 
$\sum_{i=1}^n a_i^{(m)}+\bigvee_{i=1}^n b_i^{(m)}<1,
$
and applying the strict version of \eqref{eq:boundary-ineq}, we obtain
\eqref{eq:boundary-ineq} by letting $m\to\infty$.

Now let $X_i\sim\mu_i^{r+}$, $i\in[n]$, be arbitrarily coupled and put
$S=\sum_{i=1}^n X_i$. Apply the endpoint-extended version of
\eqref{eq:boundary-ineq} with
$$
C=\frac{s}{1-r},\qquad
a_i=\frac{\alpha_i}{1-r},\qquad
b_i=\frac{\beta_i}{1-r},\quad i\in[n].
$$
The assumptions of the theorem imply
$$
\sum_{i=1}^n a_i\le C\le1,\qquad
\sum_{i=1}^n a_i+\bigvee_{i=1}^n b_i\le1.
$$
Hence
\begin{align*}
R_{\left[0,\frac{s}{1-r}\right]}(S)
\le
\sum_{i=1}^n
\Bigg[
&\frac{1-r-\beta_i}{s}
R_{\left[0,\frac{\alpha_i}{1-r}\right]
\cup
\left[\frac{\alpha_i+\beta_i}{1-r},1\right]}(X_i)\\
&+
\left(1-\frac{1-r-\beta_i}{s}\right)
R_{\left[\frac{\alpha_i}{1-r},\frac{\alpha_i+\beta_i}{1-r}\right]}(X_i)
\Bigg].
\end{align*}
Since $X_i\sim\mu_i^{r+}$, for every finite union of intervals
$I\subseteq[0,1]$, we have 
$
R_I(X_i)=R_{r+(1-r)I}(\mu_i).
$ 
Therefore,
\begin{align*}
R_{\left[0,\frac{s}{1-r}\right]}(S)
\le
\sum_{i=1}^n
\Bigg[
&\frac{1-r-\beta_i}{s}
R_{[r,r+\alpha_i]\cup[r+\alpha_i+\beta_i,1]}(\mu_i)\\
&+
\left(1-\frac{1-r-\beta_i}{s}\right)
R_{[r+\alpha_i,r+\alpha_i+\beta_i]}(\mu_i)
\Bigg].
\end{align*}
Taking the supremum over all couplings of $X_i\sim\mu_i^{r+}$ and using
Lemma \ref{Le:transform}, we obtain
\begin{align*}
\sup_{\nu\in\Lambda(\boldsymbol{\mu})}R_{[r,r+s]}(\nu)
\le
\sum_{i=1}^n
\Bigg[
&\frac{1-r-\beta_i}{s}
R_{[r,r+\alpha_i]\cup[r+\alpha_i+\beta_i,1]}(\mu_i)\\
&+
\left(1-\frac{1-r-\beta_i}{s}\right)
R_{[r+\alpha_i,r+\alpha_i+\beta_i]}(\mu_i)
\Bigg].
\end{align*}
This proves \eqref{eqn: GeneralisedNewRVaR} for
$\boldsymbol{\mu}\in\mathcal M_1^n$ and $r>0$.

Next, let $\boldsymbol{\mu}\in\mathcal M^n$ and $r>0$. Let
$X_i\sim\mu_i^{r+}$, $i\in[n]$, be arbitrarily coupled and set
$S=\sum_{i=1}^n X_i$. Since $r>0$, each $X_i$ is bounded from below.
For $m\ge1$, define
$$
X_i^{(m)}:=X_i\wedge m,\qquad i\in[n],
\qquad
S^{(m)}:=\sum_{i=1}^n X_i^{(m)}.
$$
Then $X_i^{(m)}\in L^1$ and $X_i^{(m)}\uparrow X_i$ as $m\to\infty$.
Applying the endpoint-extended auxiliary inequality \eqref{eq:boundary-ineq} to
$X_i^{(m)}$, $i\in[n]$, with
$
C=\frac{s}{1-r},\;
a_i=\frac{\alpha_i}{1-r},\;
b_i=\frac{\beta_i}{1-r},
$
gives
\begin{align*}
R_{\left[0,\frac{s}{1-r}\right]}(S^{(m)})
\le
\sum_{i=1}^n
\Bigg[
&\frac{1-r-\beta_i}{s}
R_{\left[0,\frac{\alpha_i}{1-r}\right]
\cup
\left[\frac{\alpha_i+\beta_i}{1-r},1\right]}(X_i^{(m)})\\
&+
\left(1-\frac{1-r-\beta_i}{s}\right)
R_{\left[\frac{\alpha_i}{1-r},\frac{\alpha_i+\beta_i}{1-r}\right]}(X_i^{(m)})
\Bigg].
\end{align*}
Since $X_i^{(m)}\uparrow X_i$ and $S^{(m)}\uparrow S$, the integrated
quantiles in the above display converge monotonically to the corresponding
integrated quantiles of $X_i$ and $S$. 
For terms with nonnegative coefficients, monotone convergence applies in
the extended sense. If
$
1-\frac{1-r-\beta_i}{s}<0,
$
then $\beta_i<1-r-s$ and hence
$
\alpha_i+\beta_i<1-r.
$
Therefore
$
(\alpha_i+\beta_i)/(1-r)<1.
$
Since $X_i$ is bounded from below, the corresponding average quantile
$
R_{[\alpha_i/(1-r),(\alpha_i+\beta_i)/(1-r)]}(X_i)
$
is finite. Hence the term with the negative coefficient converges in the
ordinary finite sense, and no undefined expression of the form
$\infty-\infty$ appears. 
Letting $m\to\infty$, we obtain the
same inequality for the arbitrary coupling of $X_i\sim\mu_i^{r+}$. Taking
the supremum over all such couplings and using Lemma \ref{Le:transform},
we obtain \eqref{eqn: GeneralisedNewRVaR} for
$\boldsymbol{\mu}\in\mathcal M^n$ and $r>0$.

Finally, let $r=0$ and $\boldsymbol{\mu}\in\mathcal M_1^n$. Apply the
endpoint-extended version of \eqref{eq:boundary-ineq} directly with
$
C=s,\; a_i=\alpha_i,\; b_i=\beta_i,\; i\in[n].
$
This gives \eqref{eqn: GeneralisedNewRVaR} for $r=0$, including the case
$s=1$. This completes the proof.
\end{proof}

%\sout{We emphasize that in \eqref{eqn: GeneralisedNewRVaR}, {\color{red}some $\RVaR$ should be understood as the limiting quantiles: If $\beta_i=0$, then $R_{[r + \alpha_i, r+ \alpha_i+\beta_i]}(\mu_i)=q_{r+\alpha_i}^+(\mu_i)$ and if $\alpha_i=0$ and $r+\alpha_i+\beta_i=1$, then $R_{[r, r+\alpha_i] \cup [r+ \alpha_i+\beta_i, 1]}(\mu_i)=q_{r}^+(\mu_i)+q_1^-(\mu_i)$.} } 

Note that in \eqref{eqn: GeneralisedNewRVaR}, it is possible that some $I$ for $R_I$ has a length of zero. The only possible case is $\alpha_i=0$ and $r+\alpha_i+\beta_i=1$, implying the length of $[r, r+\alpha_i] \cup [r+ \alpha_i+\beta_i, 1]$ is zero and $\frac{1-r-\beta_i}{s}=0$. Hence, the value of
$ \frac{1-r-\beta_i}{s} R_{[r, r+\alpha_i] \cup [r+ \alpha_i+\beta_i, 1]} \left(\mu_i \right)$ is understood as zero.

In Theorem 1 of \cite{ELW18}, it shows that for  $\alpha_1, \cdots, \alpha_n\geq 0$ and $\beta_1,\cdots, \beta_n>0$ satisfying $\sum_{i=1}^n \alpha_i + \bigvee_{i=1}^n \beta_i<1$, $\boldsymbol{\mu} \in \M^n$ and $\nu\in \Lambda (\boldsymbol \mu)$, we have 
    \begin{equation}\label{eq:RVaR18}
        R_{[r, r + s]}(\nu) 
        \leq 
        \sum_{i=1}^n R_{[1-\alpha_i-\beta_i, 1-\alpha_i]} \left(\mu_i\right),
    \end{equation}
    where $r+s=1-\sum_{i=1}^n \alpha_i$ and $s=\bigvee_{i=1}^n \beta_i$.  
Clearly, the individual and aggregate risk measures in \eqref{eq:RVaR18} are all $\RVaR$. However, the individual risk measures in the new $\RVaR$ inequality in \eqref{eqn: GeneralisedNewRVaR} are more complicated,  involving the linear combinations of $R_I$ with $I$ being a union of two intervals. The aggregate risk measure still has the form of $\RVaR$.  This new $\RVaR$ inequality helps establish new risk aggregation bounds and sharpness conditions. It is also useful to  investigate risk sharing for different distortion risk measures with more complex distortion functions than that of \cite{ELW18}. 

Further, by setting specific values (e.g., $1-r-\beta_i=s$) for the bound in Theorem \ref{thm: newRVaRinequality}, we arrive at a simplified upper bound displayed in the following corollary. 
\begin{corollary}\label{cor:simple}
    \label{coro: reducedRVaR} Let  $0<r < r+s \leq 1$ and $\alpha_1,\cdots,\alpha_n \in (0, 1-r)$ with $\sum_{i=1}^n \alpha_i = s$.
    Then for $\boldsymbol{\mu} \in \M^n$ and $\nu\in \Lambda (\boldsymbol \mu)$, we have 
    \begin{equation}
   \label{eq:newineq}
        R_{[r, r+s]}(\nu) \leq \sum_{i=1}^n R_{[r, r+\alpha_i] \cup [1-s+\alpha_i, 1]}(\mu_i).
    \end{equation}
    Moreover, \eqref{eq:newineq}  holds for $r=0$ if $\boldsymbol{\mu} \in \M_1^n$.
\end{corollary}
\begin{proof}
Let $X_i\sim \mu_i$, $i\in[n]$, and set $S=\sum_{i=1}^n X_i$, so that
$S\sim \nu$. 
First suppose that $r+s<1$. Apply Theorem \ref{thm: newRVaRinequality}
with
$
\beta_i=1-r-s,\; i\in[n].
$
Since $\sum_{i=1}^n\alpha_i=s$, we have
$
\frac{1-r-\beta_i}{s}=1,\; i\in[n],
$
and hence the second term on the right-hand side of
\eqref{eqn: GeneralisedNewRVaR} vanishes. 
Therefore,
$$
R_{[r,r+s]}(S)
\le
\sum_{i=1}^n
R_{[r,r+\alpha_i]\cup[r+\alpha_i+1-r-s,1]}(X_i)
=
\sum_{i=1}^n
R_{[r,r+\alpha_i]\cup[1-s+\alpha_i,1]}(X_i),
$$
which gives \eqref{eq:newineq}. This argument also covers the case
$r=0$ and $s<1$, provided that $\boldsymbol{\mu}\in\mathcal M_1^n$, by the
$r=0$ part of Theorem \ref{thm: newRVaRinequality}.

It remains to consider the endpoint case $r+s=1$. In this case,
$
[r,r+\alpha_i]\cup[1-s+\alpha_i,1]
=
[r,r+\alpha_i]\cup[r+\alpha_i,1]
=
[r,1].
$
Hence \eqref{eq:newineq} reduces to
$
R_{[r,1]}\left(\sum_{i=1}^n X_i\right)
\le
\sum_{i=1}^n R_{[r,1]}(X_i).
$ 
If $r>0$, this is the subadditivity of $\ES_{1-r}$. If $r=0$, then
$s=1$ and $R_{[0,1]}$ is the expectation on $\mathcal M_1$; hence the
inequality holds as equality:
$ 
R_{[0,1]}\left(\sum_{i=1}^n X_i\right)
=
\sum_{i=1}^n R_{[0,1]}(X_i).
$ 
This completes the proof.
\end{proof}
By setting $r=0$, Corollary \ref{coro: reducedRVaR} states that the lower tail of the aggregate risk can be controlled by the summation of the lower tail  and the upper tail  of the individual risks.
Later, we will see that the inequality in Corollary \ref{coro: reducedRVaR} plays a crucial role to solve the risk sharing problem in Section \ref{sec:RiskSharing}.

\section{Extended Convolution Bounds}
\label{sec:aggregation}
In this section, we obtain some bounds for the risk aggregation with dependence uncertainty for $\RVaR$, the difference between two $\RVaR$s, and the difference between two quantiles. The bound for $\RVaR$ is a complement to the results in \cite{BLLW24} by providing a different form and more sharpness results. The bounds for the difference between two $\RVaR$s and the difference between two quantiles are new to the literature. We call those bounds extended convolution bounds.
\subsection{RVaR Aggregation Upper and Lower Bounds}\label{sec:convolution1}

Now we show the first extended convolution bound. 
\begin{theorem}[RVaR aggregation upper bound]\label{th:ra}
	For either $\boldsymbol{\mu} \in \M^n$ and $0< r < r+s \leq 1$, or $\boldsymbol{\mu} \in \M_1^n$ and $0\leq r < r+s \leq 1$, we have %\com{YL: Here $d = r+s$.}
	\begin{equation}\label{eq:ra}
	\begin{aligned}
	\sup_{\nu\in \Lambda(\boldsymbol{\mu})}R_{[r, r+s]}(\nu) 
	\leq  \inf_{\substack{\boldsymbol{\beta} \in (1-r)\Delta_n \\\beta_0\ge 1-r-s}} \bigg\{& \left(1-\frac{1-r-\beta_0}{s}\right) \sum_{i=1}^n R_{[r+\beta_i, r+\beta_i+\beta_0]}(\mu_i) \\
	&\quad\quad + \frac{1-r-\beta_0}{s} \sum_{i=1}^n R_{[r,r+\beta_i]\cup[r+\beta_i+\beta_0, 1]}(\mu_i)\bigg\}.
	\end{aligned}
	\end{equation} 
	Moreover, \eqref{eq:ra} holds as an equality for $\boldsymbol{\mu} \in \M_1^n$ in the following cases:
	\begin{enumerate}[(i)]
		%%\item $n\le2$;
		\item \label{item:incr} 
		each of $\mu_1,\dots,\mu_n$ admits an increasing density beyond its $r$-quantile;
		%\item $r+s=1$ and each of $\mu_1,\dots,\mu_n$ admits a decreasing density beyond its $r$-quantile;
		\item $\sum_{i=1}^n \mu_{i} \left[q_r^+(\mu_{i}), q_{1}^-(\mu_{i})\right) \le  1-r$. 
	\end{enumerate}   
\end{theorem}
\begin{proof}
In light of  \eqref{eqn: GeneralisedNewRVaR} and using the assumption $\sum_{i=1}^n \alpha_i+\bigvee_{i=1}^n\beta_i=1-r$ and the transformations $\beta_i\to \beta_0,~i\in [n]$, $\alpha_i\to\beta_i,~i\in [n]$, we have, for $\boldsymbol{\beta} \in (1-r)\Delta_n$ and $\beta_0\ge 1-r-s$,
$$R_{[r, r+s]}(\nu) 
	\leq \left(1-\frac{1-r-\beta_0}{s}\right) \sum_{i=1}^n R_{[r+\beta_i, r+\beta_i+\beta_0]}(\mu_i) 
	+ \frac{1-r-\beta_0}{s} \sum_{i=1}^n R_{[r,r+\beta_i]\cup[r+\beta_i+\beta_0, 1]}(\mu_i).$$
We obtain \eqref{eq:ra} by taking the supremum on the left-hand side of the above inequality over $\nu\in \Lambda(\boldsymbol{\mu})$ and infimum on the right-hand side of the above inequality over $\boldsymbol{\beta} \in (1-r)\Delta_n$ and $\beta_0\ge 1-r-s$.
	%We prove \eqref{eq:ra} by taking the supremum of all  on the left-hand side and the infimum of all feasible $\boldsymbol{\beta}$ on the right-hand side of \eqref{eqn: GeneralisedNewRVaR} \tcrd{(\eqref{eq:ineqfull})}. 

If $r+s=1$, then $R_{[r,1]}=\mathrm{ES}_{1-r}$. By the subadditivity
of Expected Shortfall,
$
R_{[r,1]}\left(\sum_{i=1}^n X_i\right)
\le
\sum_{i=1}^n R_{[r,1]}(X_i),
$
and equality is attained by the comonotonic coupling. Hence
$
\sup_{\nu\in\Lambda(\mu)}R_{[r,1]}(\nu)
=
\sum_{i=1}^n R_{[r,1]}(\mu_i).
$ 
Moreover, the right-hand side of \eqref{eq:ra} converges to the same value
by taking $\beta_0\downarrow0$. Therefore \eqref{eq:ra} is an equality in
the endpoint case $r+s=1$. In the rest of the proof, assume $r+s<1$.

Next, we show that \eqref{eq:ra} is an equality for cases (i) and (ii).  If case (i) holds, then each of $\mu_1^{r+},\dots,\mu_n^{r+}$ admits an increasing density. Define (Equation (3.4) in \cite{JHW16})
$$T_{s_n}=h(U)\id_{\{U \leq s_n\}}+d(s_n)\id_{\{U>s_n\}},$$
where $U\sim \mathrm{U}[0,1]$,  $h(x)=\sum_{i=1}^ny_i(x)-(n-1)y(x)$, $d(x)=-\frac{1}{1-x}\int^{y(x)-y_i(x)}_{-y_i(x)}z\mu_i^{r+}(\d z)$ for $x\in (0,1)$, and $s_n=\inf\{x\in (0,1): h(x)\leq d(x)\}$ with $y, y_i$, $i\in [n]$ being continuous functions on $(0,1)$ satisfying 
$$\sum_{i=1}^n\mu_i^{r+}(-\infty,-y_i(x))=x,~~\text{and}~~\mu_i^{r+}[-y_i(x),y(x)-y_i(x))=1-x,~x\in (0,1).$$
In light of Lemma 3.4 of \cite{JHW16} and using the fact that each of $\mu_1^{r+},\dots,\mu_n^{r+}$ admits an increasing density, there exist $X_i\sim \mu_i^{r+}$ such that $T_{s_n}=\sum_{i=1}^n -X_i$. Moreover, using Lemma 3.3 of \cite{JHW16}, we could find $\boldsymbol{\beta}' \in \Delta_n$ with $\beta_0'\ge 1-\frac{s}{1-r}$  such that 
\begin{align*} R_{[0, 1-\frac{s}{1-r}]}(-S)=\sum_{i=1}^n R_{[1-\beta_i' - \beta_0', 1 - \beta_i']} (-X_i)
\end{align*}
with $S=\sum_{i=1}^n X_i$. The detail of construction of such $\boldsymbol{\beta}'$ and $\beta_0'$ is omitted as it is only involving tedious computation using Lemma 3.3 of \cite{JHW16}. One can refer to the proof of Theorem 1 of \cite{BLLW24} for the similar computation.
Note that the above equation can be rewritten as 
\begin{align*} R_{[\frac{s}{1-r},1]}(S)=\sum_{i=1}^n R_{[\beta_i',\beta_i'+\beta_0']} (X_i).
\end{align*}
Using the fact 
$$\E(S)=\frac{s}{1-r}R_{[0,\frac{s}{1-r}]}(S)+\frac{1-r-s}{1-r}R_{[\frac{s}{1-r},1]}(S),$$
we have
\begin{align*}
    R_{[0,\frac{s}{1-r}]}(S)&=\frac{1-r}{s}\left(\E(S)-\frac{1-r-s}{1-r}\sum_{i=1}^n R_{[\beta_i',\beta_i'+\beta_0']} (X_i)\right)\\
    &=\frac{1-r}{s}\sum_{i=1}^n\left((1-\beta_0')R_{[0,\beta_i']\cup [\beta_i'+\beta_0',1]}(X_i)+\left(\beta_0'-\frac{1-r-s}{1-r}\right) R_{[\beta_i',\beta_i'+\beta_0']} (X_i)\right).
\end{align*}
There exist $U'\sim \mathrm{U}[0,1]$ and  $(Y_1,\dots,Y_n)$ such that   $(Y_1,\dots,Y_n)$ is independent of $U'$ and has the same distribution as $(X_1,\dots, X_n)$. Define 
$$X_i'=Y_i\id_{\{U' \geq r\}}+q_{U'}^-(\mu_i)\id_{\{U'<r\}},~i\in[n],~ \text{and}~S'=\sum_{i=1}^n X_i'.$$
Clearly, $X_i'\sim \mu_i$ and $S'=(\sum_{i=1}^nY_i)\id_{\{U' \geq r\}}+(\sum_{i=1}^n q_{U'}^-(\mu_i))\id_{\{U'<r\}}$. Note that $q_u^-(S') = q_{\frac{u-r}{1-r}}^-(S)$ for any $u\in(r,1)
$ and $R_{[r,r+s]}(S')= R_{[0,\frac{s}{1-r}]}(S)$. Letting  $\beta_0=(1-r)\beta_0'$ and $\beta_i=(1-r)\beta_i',~i\in [n]$,  we have   $\boldsymbol{\beta} \in (1-r)\Delta_n$, $\beta_0\ge 1-r-s$ and 
\begin{align*} R_{[r,r+s]}(S')&=\frac{1-r}{s}\sum_{i=1}^n\left(\left(1-\frac{\beta_0}{1-r}\right)R_{[r,r+\beta_i]\cup [r+\beta_i+\beta_0,1]}(X_i')\right.\\
&~~\left.+\left(\frac{\beta_0}{1-r}-\frac{1-r-s}{1-r}\right) R_{[r+\beta_i,r+\beta_i+\beta_0]} (X_i')\right)\\
&=\left(1-\frac{1-r-\beta_0}{s}\right) \sum_{i=1}^n R_{[r+\beta_i, r+\beta_i+\beta_0]}(\mu_i) 
	+ \frac{1-r-\beta_0}{s} \sum_{i=1}^n R_{[r,r+\beta_i]\cup[r+\beta_i+\beta_0, 1]}(\mu_i),
\end{align*}
 implying the inverse inequality of \eqref{eq:ra}. Hence, \eqref{eq:ra} holds as an equality.

   % Note that \eqref{eq:ra} holds as an equality if and only if there exist $\boldsymbol{\beta} \in (1-r)\Delta_n$, $\beta_0\ge 1-r-s$ and $X_i\sim \mu_i$ such that 
    %\begin{align}\label{Eq:dependence}
    %R_{[r, r+s]}(S)=\left(1-\frac{1-r-\beta_0}{s}\right) \sum_{i=1}^n R_{[r+\beta_i, r+\beta_i+\beta_0]}(\mu_i) 
	%+ \frac{1-r-\beta_0}{s} \sum_{i=1}^n R_{[r,r+\beta_i)\cup(r+\beta_i+\beta_0, 1]}(\mu_i)\end{align}
    %with $S=\sum_{i=1}^n X_i$. It follows from the proof of Theorem \ref{thm: newRVaRinequality} that \eqref{Eq:dependence} holds if and only if there exist $\boldsymbol{\beta}' \in \Delta_n$, $\beta_0'\ge 1-\frac{s}{1-r}$ and $X_i\sim \mu_i^{r+}$ such that 
%\begin{align*} R_{[0, 1-s/(1-r)]}(-S)=\sum_{i=1}^n R_{[1-\beta_i' - \beta_0', 1 - \beta_i']} (-X_i),
%\end{align*}
%which is also equivalent to  $$\sup_{X_i\sim \mu_i^{r+}}R_{[0, 1-s/(1-r)]}(-S)=\inf_{\boldsymbol{\beta}' \in \Delta_n, \beta_0'\ge 1-\frac{s}{1-r}}\sum_{i=1}^n R_{[1-\beta_i' - \beta_0', 1 - \beta_i']} (-X_i).$$
% By Theorem 1 of \cite{BLLW24}, the above holds true if and only if 
%	\begin{enumerate}[(i)]
		%\item {\color{red}$r+s=1.$}
%		%%\item $n\le2$;
%		\item 
%		each of $\mu_1^{r+},\dots,\mu_n^{r+}$ admits an increasing density;
		%\item $p=1$ and each of $\mu_1^{r+},\dots,\mu_n^{r+}$ admits a decreasing density;
%		\item $\sum_{i=1}^n \mu_{i}^{r+} \left[q_0^+(\mu_{i}^{r+}), q_{1}^-(\mu_{i}^{r+})\right) \le  1$. 
%	\end{enumerate}
	%We complete the proof. 

If  case (ii) holds, then $\sum_{i=1}^n \mu_{i}^{r+} \left[q_0^+(\mu_{i}^{r+}), q_{1}^-(\mu_{i}^{r+})\right) \le  1$. Hence, there exist $X_i\sim \mu_i^{r+},~i\in [n]$ such that $-X_i,~i\in [n]$  are lower mutually exclusive. Hence, in light of Lemma EC.2 of \cite{BLLW24}, there exist $X_i\sim \mu_i^{r+}$  and $\boldsymbol{\beta}' \in \Delta_n$, $\beta_0'\ge 1-\frac{s}{1-r}$  such that 
\begin{align*} R_{[0, 1-\frac{s}{1-r}]}(-S)=\sum_{i=1}^n R_{[1-\beta_i' - \beta_0', 1 - \beta_i']} (-X_i)
\end{align*}
with $S=\sum_{i=1}^n X_i$. Repeating the procedure in the proof of case (i), we can show that \eqref{eq:ra} holds as an equality.
\end{proof}
As a comparison, Theorem 1 of \cite{BLLW24} gives the following (upper) convolution bound 
\begin{equation}\label{eq:convo1}
\sup_{\nu\in \Lambda(\boldsymbol \mu) }R_{[r, r+s]} (\nu) \le
\inf_{\substack{\boldsymbol \beta\in (1-r)\Delta_n \\\beta_0\ge s}} 
\sum_{i=1}^n R_{[1-\beta_i-\beta_0, 1-\beta_i]}(\mu_i).
\end{equation}
As stated in \cite{BLLW24}, the bound \eqref{eq:convo1} is sharp if each of $\mu_1,\dots,\mu_n$ admits a decreasing density beyond its $r$-quantile. However, \cite{BLLW24} also proposed a numerical example that the bound \eqref{eq:convo1} may not be sharp if each of $\mu_1,\dots,\mu_n$ admits an increasing density beyond its $r$-quantile.  Here, our Theorem \ref{th:ra} complements the result in \cite{BLLW24} by offering a new bound and the sharpness condition (marginals with increasing densities on their upper-tail parts). In practice, although most of the distributions have decreasing densities on their upper-tail parts, some distributions with increasing densities on the upper-tail parts still exist such as generalized Pareto distribution with shape parameter smaller than $-1$,  Beta distribution or triangular distribution. Moreover, the result in Theorem \ref{th:ra} is crucial to find the sharp bound for the best-case scenario when the marginals have decreasing densities; see Theorem \ref{th:ral} below. 

Before proceeding to the lower bound, we first obtain a simplified sharp bound if the marginal distributions are homogeneous. For $\mu\in\M$, let 
$$
c_n(\mu)=\inf\left\{x\in \left(0,\frac{1}{n}\right): \frac{(n-1)q_{(n-1)x}^{-}(\mu)+q_{1-x}^{-}(\mu)}{n}\leq R_{[(n-1)x,1-x]}(\mu)\right\}
$$
with the convention that $\inf\emptyset=\frac{1}{n}$ and define $\Lambda_n(\mu) = \Lambda(\mu, \dots, \mu)$.
For $\mu\in\M$, we denote by $\tilde{\mu}$ the distribution measure of $-X$ with $X\sim \mu$.
\begin{proposition}[RVaR aggregation upper bound: homogeneous marginal]\label{prop:ra_homo}
If $\mu\in \M$ has an increasing density on its support, %and we denote by $\tilde{\mu}$ the distribution measure of $-X_i$ where $X_i\sim \mu$, 
and $0< r < r+s \leq 1$ with $\frac{s}{1-r}\leq nc_n(\tilde{\mu}^{(1-r)-})$,  then we have
\begin{equation}\label{eq:rvar_inc}
\begin{aligned}
\sup_{\nu  \in \Lambda_n(\mu)} R_{[r, r+s]}(\nu)
&= n R_{[r, r+\frac{s}{n}] \cup [1-s+\frac{s}{n}, 1]}(\mu).
\end{aligned}
\end{equation}
\end{proposition}
\begin{proof}
 Suppose $\mu_1 = \cdots = \mu_n = \mu$, which has an increasing density on its support. Hence, $\tilde{\mu}$ has a decreasing density. Denote by $\tilde{\nu}$ the distribution measure of $-\sum_{i=1}^n X_i$, where $\sum_{i=1}^n X_i \sim \nu$. For $0 \leq r < r+s \leq 1$, write $t = 1-r-s$. Using Lemma \ref{Le:transform}, we have
	\begin{equation}\label{eq:convert1}
	\begin{aligned}
	\sup_{\nu  \in \Lambda_n(\mu)} R_{[r, r+s]}(\nu)= -\inf_{\tilde{\nu} \in \Lambda_n(\tilde{\mu})} R_{[1-r-s,1-r]}(\tilde{\nu}) = -\inf_{\tilde{\nu} \in \Lambda_n(\tilde{\mu}^{(t+s)-})} \ES_{\frac{s}{t+s}}(\tilde{\nu}).
	\end{aligned}
	\end{equation}
	Note that  $\tilde{\mu}^{(t+s)-}$ has a decreasing density on its support. 
	Hence, by Theorem 5.2 of \cite{BJW14},  we have for any $p\in (0, n c_n(\tilde{\mu}^{(t+s)-})]$,
	\begin{equation}\label{eq:inf_ES}
	\inf_{\tilde{\nu}\in \Lambda_n(\tilde{\mu}^{(t+s)-})} \ES_{p}(\tilde{\nu}) = \frac{n}{p}\int_{0}^{\frac{p}{n}} \left( (n-1)q_{(n-1)u}^- (\tilde{\mu}^{(t+s)-}) + q_{1-u}^- (\tilde{\mu}^{(t+s)-}) \right) \d u.
	\end{equation}
	Using \eqref{eq:inf_ES}, for $\frac{s}{t+s}\leq n c_n(\tilde{\mu}^{(t+s)-}) = n c_n(\tilde{\mu}^{(1-r)-})$, we have
	\begin{equation}
	\label{eq:convert2}
	\begin{aligned}
	-\inf_{\tilde{\nu} \in \Lambda_n(\tilde{\mu}^{(t+s)-})} \ES_{\frac{s}{t+s}}(\tilde{\nu}) 
	&= -\frac{n(t+s)}{s}\int_{0}^{\frac{s}{n(t+s)}}\left( (n-1)q_{(n-1)u}^{-} (\tilde{\mu}^{(t+s)-}) + q_{1-u}^{-} (\tilde{\mu}^{(t+s)-}) \right) \d u \\
	&= -\frac{n(t+s)}{s}\int_{0}^{\frac{s}{n(t+s)}}\left( (n-1)q_{(n-1)(t+s)u}^{-} (\tilde{\mu}) + q_{(1-u)(t+s)}^{-} (\tilde{\mu}) \right) \d u\\
	&= -\frac{n(t+s)}{s}\left( \int_{0}^{\frac{(n-1)}{n}s} \frac{1}{t+s} q_{v}^{-} (\tilde{\mu}) \d v + \int_{t+s-\frac{s}{n}}^{t+s} \frac{1}{t+s} q_{v}^{-} (\tilde{\mu}) \d v \right) \\
	&= -\frac{n}{s}\left( \int_{0}^{\frac{(n-1)}{n}s}  q_{v}^{-} (\tilde{\mu}) \d v + \int_{t+s-\frac{s}{n}}^{t+s} q_{v}^{-} (\tilde{\mu}) \d v \right) \\
	&= \frac{n}{s}\left( \int_{1-t-s}^{1-t-s+\frac{s}{n}} q_{v}^{-} (\mu)  \d v + \int_{1-\frac{(n-1)}{n}s}^{1}  q_{v}^{-} (\mu) \d v \right)\\
	&= n R_{[r, r+\frac{s}{n}] \cup [1-s+\frac{s}{n}, 1]}(\mu).
	\end{aligned}
	\end{equation}
	Combining \eqref{eq:convert1} and \eqref{eq:convert2}, we have \eqref{eq:rvar_inc}.
\end{proof}

Based on the result in Theorem \ref{th:ra},  we immediately obtain a  lower bound for the risk aggregation of $\RVaR$ with fixed marginal distributions but unknown dependence structure. 
\begin{theorem}[$\RVaR$ aggregation lower bound]\label{th:ral}
	For either $\boldsymbol{\mu} \in \M^n$ and $0\leq r < r+s<1$, or $\boldsymbol{\mu} \in \M_1^n$ and $0\leq r<r+s\leq 1$, we have %\com{YL: Here $d = r+s$.}
	\begin{equation}\label{eq:ral}
	\begin{aligned}
	\inf_{\nu\in \Lambda(\boldsymbol{\mu})}R_{[r, r+s]}(\nu) 
	\geq  \sup_{\substack{\boldsymbol{\beta} \in (r+s)\Delta_n \\\beta_0\ge r}} \bigg\{& \left(1-\frac{r+s-\beta_0}{s}\right) \sum_{i=1}^n R_{[r+s-\beta_i-\beta_0, r+s-\beta_i]}(\mu_i) \\
	&+ \frac{r+s-\beta_0}{s} \sum_{i=1}^n R_{[0,r+s-\beta_i-\beta_0] \cup [r+s-\beta_i, r+s]}(\mu_i)\bigg\}.
	\end{aligned}
	\end{equation} 
	Moreover, \eqref{eq:ral} holds as an equality for $\boldsymbol{\mu} \in \M_1^n$ in the following cases:
	\begin{enumerate}[(i)]
		%%\item $n\le2$;
		\item \label{item:incr} 
		each of $\mu_1,\dots,\mu_n$ admits a decreasing density below its $(r+s)$-quantile;
		%\item $r+s=1$ and each of $\mu_1,\dots,\mu_n$ admits a decreasing density beyond its $r$-quantile;
		\item $\sum_{i=1}^n \mu_{i} \left(q_0^+(\mu_{i}), q_{r+s}^-(\mu_{i})\right] \le r+s$. 
	\end{enumerate}   
\end{theorem}
\begin{proof}
    Denote by $\tilde{\mu}_i$ the distribution measure of $-X_i$, where $X_i\sim \mu_i$, and  by $\tilde{\nu}$ the distribution measure of $-\sum_{i=1}^n X_i$, where $\sum_{i=1}^n X_i \sim \nu$. Then we have
    $$\inf_{\nu\in \Lambda(\boldsymbol{\mu})}R_{[r, r+s]}(\nu) =-\sup_{\tilde{\nu}\in \Lambda(\tilde{\boldsymbol{\mu}})}R_{[1-r-s, 1-r]}(\tilde{\nu}).$$
    Applying Theorem \ref{th:ra}, we immediately establish the claim.
\end{proof}
It is worth mentioning that Theorem A.1 of \cite{BLLW24} gives the following (lower) convolution bound
\begin{equation}\label{eq:convo11}
\inf_{\nu\in \Lambda(\boldsymbol \mu) }R_{[r, r+s]} (\nu) \geq
\sup_{\substack{\boldsymbol \beta\in (r+s)\Delta_n \\\beta_0\ge s}} 
\sum_{i=1}^n R_{[\beta_i,\beta_i+\beta_0]}(\mu_i).
\end{equation}
It is shown in \cite{BLLW24} that \eqref{eq:convo11} is a sharp bound if each marginal distribution admits an increasing density below its ($r+s$)-quantile. However, it may be not  sharp for the case with decreasing densities on the lower-tail parts; see the numerical example in \cite{BLLW24}. Our Theorem \ref{th:ral}  fills in this gap by providing a bound which is sharp for the best case of $\RVaR$ if all marginal distributions admit decreasing densities on their lower-tail parts. This gap is actually significant because many commonly used distributions in finance and risk management have decreasing densities on their lower-tail parts, including exponential, Pareto, and some gamma and chi distributions.  It is reasonable to expect  that the bounds in Theorems \ref{th:ra} and \ref{th:ral} approximate the exact values whenever the marginals are close to satisfying the sharpness conditions. Moreover,  our results in Theorems \ref{th:ra} and \ref{th:ral}  are the building blocks to consider the worst-case value of the difference between two $\RVaR$s in Section \ref{Sec:Robust}.

Applying Proposition \ref{prop:ra_homo}, we immediately obtain a simplified sharp lower bound if the marginal distributions are homogeneous as below. 
\begin{proposition}[RVaR aggregation lower bound: homogeneous marginal]\label{prop:ra_homo1}
	 If  $\mu\in \M$ has a decreasing density on its support, and  $0\leq r < r+s<1$ with $\frac{s}{r+s}\leq nc_n(\mu^{(r+s)-})$,  then we have
	\begin{equation*}
	\begin{aligned}
	\inf_{\nu  \in \Lambda_n(\mu)} R_{[r, r+s]}(\nu)
	&= n R_{[0,\frac{(n-1)s}{n}] \cup [r+\frac{(n-1)s}{n}, r+s]}(\mu).
	\end{aligned}
	\end{equation*}
\end{proposition}

% {\color{red}\begin{corollary}[Extended convolution bounds for VaR]
    
% \end{corollary}
% }
% =============================================================================
% Inter-RVaR difference
% =============================================================================
    
\subsection{Inter-RVaR Difference Aggregation Upper Bounds}\label{Sec:Robust}

For $0\leq r_1<s_1\leq r_2<s_2\leq 1$, the inter-$\RVaR$ difference is defined as 
$$\mathrm{IRD}_{[r_1,s_1],[r_2,s_2]}(\mu)=R_{[r_2,s_2]}(\mu)-R_{[r_1,s_1]}(\mu),$$
which can be viewed as an example of the variability measures introduced in \cite{BFWW22}.
If $r_1=1-s_2=0$ and $s_1=1-r_2$, then $\mathrm{IRD}_{[r_1,s_1],[r_2,s_2]}(\mu)=\mathrm{ES}_{1-r_2}(\mu)-\mathrm{LES}_{1-r_2}(\mu)$, which is the inter-$\ES$ difference introduced by \cite{BFWW22}. If $s_2\downarrow r_2=1-r$ and $r_1\uparrow s_1=r$, then $\mathrm{IRD}_{[r_1,s_1],[r_2,s_2]}$ converges to the inter-quantile difference
 $$\mathrm{IQD}_{r}^+(\mu)=q_{1-r}^+(\mu)-q_{r}^-(\mu),~r\in \left(0, \frac{1}{2}\right],$$
 which was given by \cite{BFWW22}. An alternative definition of inter-quantile difference was introduced by \cite{GLW23} in the study of risk sharing, i.e., 
 $$\mathrm{IQD}_{r}^-(\mu)=q_{1-r}^-(\mu)-q_{r}^+(\mu),~r\in \left[0, \frac{1}{2}\right),$$
 which is the limit of $\mathrm{IRD}_{[r_1,s_1],[r_2,s_2]}(\mu)$ as $r_2\uparrow s_2=1-r$ and $s_1\downarrow r_1=r$.
 Note that those risk measures can be used to evaluate the variability of the financial losses or risk as an alternative to variance.
Before discussing the robust risk aggregation for $\mathrm{IRD}$, we introduce some concepts and obtain a general result on the robust risk aggregation of the difference between two tail risk measures.

For $p\in (0,1)$, we say $\rho:\M\to\R\cup\{\infty\}$ is a \emph{$p$-upper-tail risk measure} if $\rho(\mu_1)=\rho(\mu_2)$ for all $\mu_1, \mu_2\in \M$ satisfying $\mu_1^{p+}=\mu_2^{p+}$; we say $\rho:\M\to\R\cup\{\infty\}$ is a \emph{$p$-lower-tail risk measure} if $\rho(\mu_1)=\rho(\mu_2)$ for all $\mu_1, \mu_2\in \M$ satisfying $\mu_1^{p-}=\mu_2^{p-}$. Note that here the domain of $\rho$ can be restricted to $\M_1$. The $p$-upper-tail risk measure is the so-called $p$-tail risk measure introduced and studied in \cite{LW21}. We say $\rho:\M\to\R\cup\{\infty\}$ satisfies \emph{monotonicity} if $\rho(\mu_1)\leq \rho(\mu_2)$ for all $\mu_1,\mu_2\in\M$ satisfying $q_{t}^-(\mu_1)\leq q_{t}^-(\mu_2)$ for all $t\in (0,1)$.

Clearly, for $0<r<s<1$,  $R_{[r,s]}$ is an $r$-upper-tail risk measure and also an  $s$-lower-tail risk measure; $q_r^+$ is an $r$-upper-tail risk measure and $(r+\epsilon)$-lower-tail risk measure for some $0<\epsilon<1-r$; $q_r^-$ is an $r$-lower-tail risk measure and $(r-\epsilon)$-upper-tail risk measure for some $0<\epsilon<r$. We refer to \cite{LW21} for properties, applications and more examples on $p$-upper-tail risk measures.
\begin{theorem}\label{Thm:TRM} For $0<r\leq s<1$, suppose that $\rho_1:\mathcal M\to\R \cup\{\infty\}$ is an $s$-upper-tail risk measure and $\rho_2:\mathcal M\to\R\cup\{\infty\}$ is an $r$-lower-tail risk measure.  If $\rho_1$ and $\rho_2$ are monotone risk measures,  for  $\boldsymbol{\mu} \in \M^n$, we have
    $$\sup_{\nu\in \Lambda(\boldsymbol{\mu})}(\rho_1(\nu)-\rho_2(\nu))=\sup_{\nu\in \Lambda(\boldsymbol{\mu})}\rho_1(\nu)-\inf_{\nu\in \Lambda(\boldsymbol{\mu})}\rho_2(\nu).$$
\end{theorem}
\begin{proof}
    Clearly, we have
     $$\sup_{\nu\in \Lambda(\boldsymbol{\mu})}(\rho_1(\nu)-\rho_2(\nu))\leq\sup_{\nu\in \Lambda(\boldsymbol{\mu})}\rho_1(\nu)-\inf_{\nu\in \Lambda(\boldsymbol{\mu})}\rho_2(\nu).$$
     Next, we show the inverse inequality. %For convenience, we write $\rho(X)=\rho(\mu)$ with $X\sim \mu$ for $\rho:\M\to\R\cup {\infty\}$.
     In light of Theorem 3 of \cite{LW21}, we have 
     $$\sup_{\nu\in \Lambda(\boldsymbol{\mu})}\rho_1(\nu)=\sup_{\nu\in \Lambda(\boldsymbol{\mu}^{s+})}\rho_1^*(\nu)=\sup_{U_i\sim \mathrm{U}[s,1]}\rho_1^*\left(\sum_{i=1}^n q_{U_i}^-(\mu_i)\right),$$
     where $\rho_1^*$ is the generator of $\rho_1$ satisfying $\rho_1(\mu)=\rho_1^*(\mu^{s+})$ for all $\mu\in\M$.
     Next, we show that
     \begin{align}\label{eq:infgenerator}
     \inf_{\nu\in \Lambda(\boldsymbol{\mu})}\rho_2(\nu)=\inf_{\nu\in \Lambda(\boldsymbol{\mu}^{r-})}\rho_2^*(\nu)=\inf_{V_i\sim \mathrm{U}[0,r]}\rho_2^*\left(\sum_{i=1}^n q_{V_i}^-(\mu_i)\right),
     \end{align}
     where  $\rho_2^*$ satisfies $\rho_2(\mu)=\rho_2^*(\mu^{r-})$ for all $\mu\in\M$. For any
$(V_1,\dots,V_n)$ with $V_i\sim \mathrm{U}[0,r]$, $i\in [n]$, there exist $(V_1',\dots,V_n')$ and $U'$ such that $(V_1',\dots,V_n')\overset{d}{=} (V_1,\dots,V_n)$ and $U'\sim \mathrm{U}[0,1]$ is independent of $(V_1',\dots,V_n')$, where $\overset{d}{=}$ means equality in distribution. Let 
$$X_i=q_{V_i'}^{-}(\mu_i)\id_{\{U' \leq r\}}+q_{U'}^{-}(\mu_i)\id_{\{U'>r\}},~i\in [n].$$
Then we have $X_i\sim\mu_i$ and 
$$\sum_{i=1}^n X_i=\left(\sum_{i=1}^n q_{V_i'}^{-}(\mu_i)\right)\id_{\{U' \leq r\}}+\left(\sum_{i=1}^nq_{U'}^{-}(\mu_i)\right)\id_{\{U'>r\}}.$$
Note that $\esssup \sum_{i=1}^n q_{V_i'}^{-}(\mu_i)\leq \sum_{i=1}^nq_{r}^{-}(\mu_i)$. Hence, we have 
$$\rho_2\left(\sum_{i=1}^n X_i\right)=\rho_2^*\left(\sum_{i=1}^n q_{V_i'}^{-}(\mu_i)\right)=\rho_2^*\left(\sum_{i=1}^n q_{V_i}^{-}(\mu_i)\right).$$
Using the arbitrariness of $V_i$, we have 
$$ \inf_{\nu\in \Lambda(\boldsymbol{\mu})}\rho_2(\nu)\leq \inf_{V_i\sim \mathrm{U}[0,r]}\rho_2^*\left(\sum_{i=1}^n q_{V_i}^-(\mu_i)\right).$$
For $X_i\sim \mu_i,~i\in [n]$, let $S=\sum_{i=1}^n X_i$. Let $(V_1,\dots, V_n)\overset{d}{=} (U_{X_1},\dots, U_{X_n})|_{U_S<r}$.
Direct computation shows $\mathbb P(V_i\leq x)=\frac{\mathbb P(U_{X_i}\leq x, U_S<r)}{r}\leq \frac{x}{r}\wedge 1$ for $x\in [0,1]$. Hence, there exist $V_i'\sim \mathrm{U}[0,r],~i\in [n]$ such that $V_i'\leq V_i,~i\in [n]$. Note that the monotonicity of $\rho_2$ implies the monotonicity of $\rho_2^*$.  Consequently, we have
$$\rho_2(S)=\rho_2^*(S|_{U_S<r})=\rho_2^*\left(\sum_{i=1}^n q_{V_i}^-(\mu_i)\right)\geq \rho_2^*\left(\sum_{i=1}^n q_{V_i'}^-(\mu_i)\right),$$
implying
$$ \inf_{\nu\in \Lambda(\boldsymbol{\mu})}\rho_2(\nu)\geq \inf_{V_i\sim \mathrm{U}[0,r]}\rho_2^*\left(\sum_{i=1}^n q_{V_i}^-(\mu_i)\right).$$
We establish the claim \eqref{eq:infgenerator}.

     For any  $(U_1,\dots, U_n)$ and $(V_1,\dots,V_n)$ with $U_i\sim \mathrm{U}[s,1]$ and $V_i\sim \mathrm{U}[0,r]$, $i\in [n]$, there exist  $(U_1',\dots, U_n')$, $(V_1',\dots,V_n')$, and $U'$ such that $(U_1',\dots, U_n', V_1',\dots,V_n')\overset{d}{=}(U_1,\dots, U_n, V_1,\dots,V_n)$ and $U'\sim \mathrm{U}[0,1]$ is independent of $(U_1',\dots, U_n', V_1',\dots,V_n')$.  Let 
$$X_i=q_{U_i'}^{-}(\mu_i)\id_{\{U'>s\}}+q_{U'}^{-}(\mu_i)\id_{\{r<U'\leq s\}}+q_{V_i'}^{-}(\mu_i)\id_{\{U'\leq r\}}.$$ 
Clearly, $X_i\sim \mu_i$ and $$\sum_{i=1}^nX_i=\left(\sum_{i=1}^n q_{U_i'}^{-}(\mu_i)\right)\id_{\{U'>s\}}+\left(\sum_{i=1}^nq_{U'}^{-}(\mu_i)\right)\id_{\{r<U'\leq s\}}+\left(\sum_{i=1}^n q_{V_i'}^{-}(\mu_i)\right)\id_{\{U'\leq r\}}.$$ Note that $\essinf \sum_{i=1}^n q_{U_i'}^{-}(\mu_i)\geq \sum_{i=1}^nq_{s}^{-}(\mu_i)\geq \sum_{i=1}^nq_{r}^{+}(\mu_i)\geq \esssup \sum_{i=1}^n q_{V_i'}^{-}(\mu_i)$.
Hence, \begin{align*}\rho_1\left(\sum_{i=1}^n X_i\right)-\rho_2\left(\sum_{i=1}^n X_i\right)&=\rho_1^*\left(\sum_{i=1}^n q_{U_i'}^{-}(\mu_i)\right)-\rho_2^*\left(\sum_{i=1}^n q_{V_i'}^{-}(\mu_i)\right)\\
&=\rho_1^*\left(\sum_{i=1}^n q_{U_i}^{-}(\mu_i)\right)-\rho_2^*\left(\sum_{i=1}^n q_{V_i}^{-}(\mu_i)\right).
\end{align*}
Using the arbitrariness of $U_i$ and $V_i$, we have
$$\sup_{\nu\in \Lambda(\boldsymbol{\mu})}(\rho_1(\nu)-\rho_2(\nu))\geq \sup_{U_i\sim \mathrm{U}[s,1]}\rho_1^*\left(\sum_{i=1}^n q_{U_i}^-(\mu_i)\right)-\inf_{V_i\sim \mathrm{U}[0,r]}\rho_2^*\left(\sum_{i=1}^n q_{V_i}^-(\mu_i)\right),$$
which shows the inverse inequality. We complete the proof.

\end{proof}

Applying Theorem \ref{Thm:TRM}, we obtain the following results.
\begin{corollary}[IRD aggregation upper bound]\label{prop:RVaR} \begin{enumerate}[(i)]
  \item 
For either $\boldsymbol{\mu} \in \M^n$ and $0<r_1<s_1\leq r_2<s_2 \leq 1$, or $\boldsymbol{\mu} \in \M_1^n$ and $0\leq r_1<s_1\leq r_2<s_2 \leq 1$, we have 
    $$\sup_{\nu\in \Lambda(\boldsymbol{\mu})}\mathrm{IRD}_{[r_1,s_1],[r_2,s_2]}(\nu)=\sup_{\nu\in \Lambda(\boldsymbol{\mu})}R_{[r_2,s_2]}(\nu)-\inf_{\nu\in \Lambda(\boldsymbol{\mu})}R_{[r_1,s_1]}(\nu).$$
 \item For $\boldsymbol{\mu} \in \M^n$ and $0<r\leq s<1$, we have $$\sup_{\nu\in \Lambda(\boldsymbol{\mu})} \left(q_{s}^+(\nu)-q_{r}^-(\nu)\right)=\sup_{\nu\in \Lambda(\boldsymbol{\mu})}q_{s}^+(\nu)-\inf_{\nu\in \Lambda(\boldsymbol{\mu})}q_{r}^-(\nu).$$
 \item For $\boldsymbol{\mu} \in \M^n$ with continuous $q_{\cdot}^-(\mu_1),\dots, q_{\cdot}^-(\mu_n)$ on $(0,1)$ and $0<r<s<1$, we have $$\sup_{\nu\in \Lambda(\boldsymbol{\mu})} \left(q_{s}^-(\nu)-q_{r}^+(\nu)\right)=\sup_{\nu\in \Lambda(\boldsymbol{\mu})} \left(q_{s}^+(\nu)-q_{r}^-(\nu)\right)=\sup_{\nu\in \Lambda(\boldsymbol{\mu})}q_{s}^+(\nu)-\inf_{\nu\in \Lambda(\boldsymbol{\mu})}q_{r}^-(\nu).$$
    \end{enumerate}
\end{corollary} 
 \begin{proof} 

The claims in (i) and (ii) follow directly from Theorem \ref{Thm:TRM}. We next focus on (iii). Note that $q_s^-$ is an $\frac{r+s}{2}$-upper-tail risk measure and $q_r^+$ is an $\frac{r+s}{2}$-lower-tail risk measure. Moreover, both $q_s^-$ and $q_r^+$ are monotone risk measures. Hence, applying Theorem \ref{Thm:TRM}, we have
$$\sup_{\nu\in \Lambda(\boldsymbol{\mu})} \left(q_{s}^-(\nu)-q_{r}^+(\nu)\right)=\sup_{\nu\in \Lambda(\boldsymbol{\mu})}q_{s}^-(\nu)-\inf_{\nu\in \Lambda(\boldsymbol{\mu})}q_{r}^+(\nu).$$
By (ii), we have
$$\sup_{\nu\in \Lambda(\boldsymbol{\mu})} \left(q_{s}^+(\nu)-q_{r}^-(\nu)\right)=\sup_{\nu\in \Lambda(\boldsymbol{\mu})}q_{s}^+(\nu)-\inf_{\nu\in \Lambda(\boldsymbol{\mu})}q_{r}^-(\nu).$$
Moreover, it follows from Lemma 4.4 of \cite{BJW14} and the continuity of $q_{t}^-(\mu_1),\dots, q_{t}^-(\mu_n)$ on $(0,1)$ that $\sup_{\nu\in \Lambda(\boldsymbol{\mu})}q_{t}^+(\nu)$ is continuous on $(0,1)$, which further implies $\inf_{\nu\in \Lambda(\boldsymbol{\mu})}q_{t}^-(\nu)$ is continuous on $(0,1)$.
Hence, we have
$$\sup_{\nu\in \Lambda(\boldsymbol{\mu})}q_{t}^+(\nu)=\sup_{\nu\in \Lambda(\boldsymbol{\mu})}q_{t}^-(\nu),~\inf_{\nu\in \Lambda(\boldsymbol{\mu})}q_{t}^+(\nu)=\inf_{\nu\in \Lambda(\boldsymbol{\mu})}q_{t}^-(\nu),~t\in (0,1).$$
Consequently,
$$\sup_{\nu\in \Lambda(\boldsymbol{\mu})} \left(q_{s}^-(\nu)-q_{r}^+(\nu)\right)=\sup_{\nu\in \Lambda(\boldsymbol{\mu})} \left(q_{s}^+(\nu)-q_{r}^-(\nu)\right).$$
 This completes the proof.
 \end{proof}

Theorem \ref{Thm:TRM} suggests that robust risk aggregation of the difference between two tail risk measures with dependence uncertainty equals the difference between two robust tail risk measures if the two regions of the tails  do not intersect. This is because  the worst-case dependence structure only concerns the dependence in the tail corner  of the marginals. %The conclusion in Theorem \ref{prop:RVaR} can be extended to the difference between two  general tail risk measures as long as the tail regions do not intersect using  Theorem 3 of  \cite{LW21} and the arguments in the proof of Theorem \ref{prop:RVaR}. 
In light of Corollary \ref{prop:RVaR}, in order to find the (sharp) bound of the robust $\mathrm{IRD}$, it suffices to find the (sharp) bound for robust $\RVaR$, which are given in Theorems 1 and A.1 of \cite{BLLW24} and Theorems \ref{th:ra} and \ref{th:ral} of this paper.

Note that (ii) of Corollary \ref{prop:RVaR} shows that for $r\in (0,1)$, the largest possible difference between $q_{r}^+(\nu)$ and $q_{r}^-(\nu)$ has the following expression:  $$\sup_{\nu\in \Lambda(\boldsymbol{\mu})} \left(q_{r}^+(\nu)-q_{r}^-(\nu)\right)=\sup_{\nu\in \Lambda(\boldsymbol{\mu})}q_{r}^+(\nu)-\inf_{\nu\in \Lambda(\boldsymbol{\mu})}q_{r}^-(\nu),$$
which is strictly positive in most cases, even if all $\mu_i$ have decreasing densities.

In light of Theorems \ref{th:ra}-\ref{th:ral} and  Corollary \ref{prop:RVaR} of the current paper and Theorems 1-2 and A.1-A.2 of \cite{BLLW24}, we immediately obtain the bounds for robust IRD and the difference between two quantiles. In what follows, we only present the sharp bounds and the corresponding conditions on the marginal distributions.

Applying  Theorem \ref{th:ral} and Corollary \ref{prop:RVaR} of this paper and Theorems 1 and A.1 of \cite{BLLW24}, we immediately arrive at the following results.
\begin{proposition}\label{Prop:r1}
    Suppose $\boldsymbol{\mu} \in \M_1^n$ and $0\leq r_1<s_1\leq r_2<s_2\leq 1$. 
    \begin{enumerate}[(i)]
    \item If each of $\mu_1,\dots,\mu_n$ admits a decreasing density beyond its $r_2$-quantile and an increasing density below its $s_1$-quantile, then   \begin{align*}\sup_{\nu\in \Lambda(\boldsymbol{\mu})}\mathrm{IRD}_{[r_1,s_1],[r_2,s_2]}(\nu)=\inf_{\substack{\boldsymbol \beta\in (1-r_2)\Delta_n \\\beta_0\ge s_2-r_2}} 
\sum_{i=1}^n R_{[1-\beta_i-\beta_0, 1-\beta_i]}(\mu_i)-\sup_{\substack{\boldsymbol \beta\in s_1\Delta_n \\\beta_0\ge s_1-r_1}} 
\sum_{i=1}^n R_{[\beta_i,\beta_i+\beta_0]}(\mu_i).
    \end{align*}
    \item If each of $\mu_1,\dots,\mu_n$ admits a decreasing density beyond its $r_2$-quantile and  below its $s_1$-quantile respectively, then   \begin{align*}\sup_{\nu\in \Lambda(\boldsymbol{\mu})}\mathrm{IRD}_{[r_1,s_1],[r_2,s_2]}(\nu)&=\inf_{\substack{\boldsymbol \beta\in (1-r_2)\Delta_n \\\beta_0\ge s_2-r_2}} 
\sum_{i=1}^n R_{[1-\beta_i-\beta_0, 1-\beta_i]}(\mu_i)\\
&\quad-\sup_{\substack{\boldsymbol{\beta} \in s_1\Delta_n \\\beta_0\ge r_1}} \bigg\{ \left(1-\frac{s_1-\beta_0}{s_1-r_1}\right) \sum_{i=1}^n R_{[s_1-\beta_i-\beta_0, s_1-\beta_i]}(\mu_i) \\
	&\quad + \frac{s_1-\beta_0}{s_1-r_1} \sum_{i=1}^n R_{[0,s_1-\beta_i-\beta_0]\cup[s_1-\beta_i, s_1]}(\mu_i)\bigg\}.
    \end{align*}
    \end{enumerate}
\end{proposition}

In light of  Theorem \ref{th:ra} and Corollary \ref{prop:RVaR} of this paper and Theorem  A.1 of \cite{BLLW24}, we obtain the following results.
\begin{proposition}\label{Prop:r2}
    Suppose $\boldsymbol{\mu} \in \M_1^n$ and $0\leq r_1<s_1\leq r_2<s_2\leq 1$. 
    \begin{enumerate}[(i)]
    \item If each of $\mu_1,\dots,\mu_n$ admits an increasing density beyond its $r_2$-quantile and below its $s_1$-quantile, then   \begin{align*}\sup_{\nu\in \Lambda(\boldsymbol{\mu})}\mathrm{IRD}_{[r_1,s_1],[r_2,s_2]}(\nu)&=\inf_{\substack{\boldsymbol{\beta} \in (1-r_2)\Delta_n \\
    \beta_0\ge 1-s_2}} \bigg\{\left(1-\frac{1-r_2-\beta_0}{s_2-r_2}\right) \sum_{i=1}^n R_{[r_2+\beta_i, r_2+\beta_i+\beta_0]}(\mu_i) \\
	&\quad+ \frac{1-r_2-\beta_0}{s_2-r_2} \sum_{i=1}^n R_{[r_2,r_2+\beta_i]\cup[r_2+\beta_i+\beta_0, 1]}(\mu_i)\bigg\}\\
    &\quad-\sup_{\substack{\boldsymbol \beta\in s_1\Delta_n \\\beta_0\ge s_1-r_1}} 
\sum_{i=1}^n R_{[\beta_i,\beta_i+\beta_0]}(\mu_i).
    \end{align*}
    \item If each of $\mu_1,\dots,\mu_n$ admits an increasing density beyond its $r_2$-quantile and  a decreasing density below its $s_1$-quantile, then   \begin{align*}\sup_{\nu\in \Lambda(\boldsymbol{\mu})}\mathrm{IRD}_{[r_1,s_1],[r_2,s_2]}(\nu)&=\inf_{\substack{\boldsymbol{\beta} \in (1-r_2)\Delta_n \\\beta_0\ge 1-s_2}} \bigg\{\left(1-\frac{1-r_2-\beta_0}{s_2-r_2}\right) \sum_{i=1}^n R_{[r_2+\beta_i, r_2+\beta_i+\beta_0]}(\mu_i) \\
	&\quad+ \frac{1-r_2-\beta_0}{s_2-r_2} \sum_{i=1}^n R_{[r_2,r_2+\beta_i]\cup[r_2+\beta_i+\beta_0, 1]}(\mu_i)\bigg\}\\
&\quad-\sup_{\substack{\boldsymbol{\beta} \in s_1\Delta_n \\\beta_0\ge r_1}} \bigg\{ \left(1-\frac{s_1-\beta_0}{s_1-r_1}\right) \sum_{i=1}^n R_{[s_1-\beta_i-\beta_0, s_1-\beta_i]}(\mu_i) \\
	&\quad + \frac{s_1-\beta_0}{s_1-r_1} \sum_{i=1}^n R_{[0,s_1-\beta_i-\beta_0]\cup[s_1-\beta_i, s_1]}(\mu_i)\bigg\}.
    \end{align*}
    \end{enumerate}
\end{proposition}

Combining Corollary \ref{prop:RVaR} of this paper with Theorems 2 and A.2 of \cite{BLLW24}, we obtain the following results.
\begin{proposition}\label{Prop:r3}
    For $\boldsymbol{\mu} \in \M^n$ and $0<r\leq s<1$, if each of $\mu_1,\dots,\mu_n$ has a density that is monotone in the same direction beyond its $s$-quantile and also monotone in the same direction  below its $r$-quantile, then we have $$\sup_{\nu\in \Lambda(\boldsymbol{\mu})} \left(q_{s}^+(\nu)-q_{r}^-(\nu)\right)=\inf_{\substack{\boldsymbol \beta\in (1-s)\Delta_n }} 
\sum_{i=1}^n R_{[1-\beta_i-\beta_0, 1-\beta_i]}(\mu_i)-\sup_{\substack{\boldsymbol \beta\in r\Delta_n}} 
\sum_{i=1}^n R_{[\beta_i,\beta_i+\beta_0]}(\mu_i).$$
\end{proposition}

The sharp bounds given in Propositions \ref{Prop:r1}-\ref{Prop:r3} are the upper bounds that are valid for all $\boldsymbol{\mu} \in \M_1^n$ and $0\leq r_1<s_1\leq r_2<s_2\leq 1$ or $\boldsymbol{\mu} \in \M^n$ and $0<r_1<s_1\leq r_2<s_2\leq 1$. Note that the sharp bounds in  (ii) of Proposition \ref{Prop:r1} and (i) and (ii) of Proposition \ref{Prop:r2} are based on the results in Theorems \ref{th:ra} and \ref{th:ral}, showing their applications to risk aggregation for other risk measures. It is also worth mentioning that based on Theorem \ref{th:ral}, (ii) of Proposition \ref{Prop:r1} gives a sharp upper bound for $\mathrm{IRD}$ if the marginals have decreasing densities on their upper-tail parts and lower-tail parts respectively. This assumption is valid for many practically used distributions in finance and risk management, such as exponential, Pareto, and some gamma and chi distributions.

The extended convolution bounds developed in this section, particularly the sharp bounds, can be applied directly in risk management, including risk evaluation under dependence uncertainty (conservative regulatory capital calculation; e.g., \cite{EPR13} and \cite{EKP20}),  portfolio optimization under dependence uncertainty (e.g., \cite{PP18}, \cite{CLLW22} and \cite{BLLW24}) and optimal insurance/reinsurance design under dependence uncertainty (e.g.,  \cite{FHLX25}). Other applications in operations research can be seen in the  discussion in Section \ref{sec:intro}.

\section{%Advances in Quantile-Based 
%Pareto-Optimal 
Risk Sharing}\label{sec:RiskSharing}

In this section, we consider the %Pareto-optimal
risk sharing problem among $n \geq 2$ agents, where their preferences are represented by some risk functionals $\rho_i:\X_1\to \R$, $i \in [n]$. For a total risk $X \in \X_1$,  the set of all possible allocations of $X$ is denoted by
$$
\A_n(X) = \left\{(X_1, \cdots, X_n) \in (\X_1)^n: \sum_{i=1}^n X_i = X \right\}.
$$
The  \emph{inf-convolution} of  $\rho_1,\dots,\rho_n$ is defined as
\begin{equation*} 
 \dsquare_{i=1}^n \rho_i (X)  = \inf\left\{\sum_{i=1}^n \rho_i(X_i): (X_1,\cdots,X_n)\in \mathbb A_n(X)  \right\},~~X\in \X_1.
\end{equation*}
%We say an allocation  $(X_1^*,\dots, X_n^*)\in \mathbb A_n(X)$ is \emph{optimal} if $\sum_{i=1}^n \rho_i(X_i^*)=\square_{i=1}^n \rho_i (X).$ 
In this section, we use risk functionals $R_{I_1},\dots, R_{I_n}$ to represent the agents' preferences, where $I_i = [0, \beta_i] \cup [1-\beta+\beta_i, 1]$ for all $i \in [n]$ with a parameter vector $\bsyb{\beta}: =(\beta_1, \cdots, \beta_n)\in (0, 1)^n$ satisfying $\beta=\sum_{i=1}^n\beta_i$ and $0 < \beta < 1$. 
%We say that $\bsyb{\beta}$ is component-wise positive if $\beta_i > 0$ for all  $i\in [n]$. 
We aim to find the \emph{optimal allocation} for the inf-convolution, i.e., finding $(X_1^*,\dots,X_n^*)\in \mathbb A_n(X)$ such that
\begin{align}\label{Eq:inf-convolution}
\sum_{i=1}^n R_{I_i%[0, \beta_i] \cup [1-\beta+\beta_i, 1]
}(X_i^*) =\dsquare_{i=1}^n R_{I_i%[0, \beta_i] \cup [1-\beta+\beta_i, 1]
}(X). 
\end{align}
%{\color{red}Throughout this section, the total risk $X$ is a fixed random variable and we denote its $\esssup$ by $t$, i.e., $t := q_{1}^-(X)$. Note that $t=\infty$ if $X$ is not bounded from above. }

The central object in our investigation of risk sharing is the form taken by $R_{[0, \alpha] \cup [\gamma, 1]}$ with $ 0 \leq \alpha < \gamma \leq 1$. %let  where we let $[n]$ denote the set $\{1,\cdots,n\}$ throughout this article. 
In the risk sharing context, the motivation for using such preference functionals can be interpreted from the fact that for $\lambda\in (0,1)$,
\begin{equation}\label{eq:dist_risk_meas}
    \lambda \mathbb E(X)+(1-\lambda) R_{I_i}(X)=\int_0^1 q_{1-s}^-(X)\d g_{\lambda,i}(s):=\rho_{g_{\lambda,i}}(X),
\end{equation}
where \begin{align}\label{DF}
g_{\lambda, i}(s)=\lambda s+\frac{1-\lambda}{\beta}(s\wedge(\beta-\beta_i))\id_{[0, 1-\beta_i]}(s)+\frac{1-\lambda}{\beta}(s-1+\beta)\id_{[1-\beta_i,1]}(s)
\end{align}
is a distortion function. Hence, the above risk functionals $\rho_{g_{\lambda,i}}$ can be viewed as either a Yaari's dual utility or a distortion risk measure; see \cite{Y87} and Chapter 4 of \cite{FS16}. Note that the distortion function $g_{\lambda,i}$ has an inverse-S shape, exaggerating the probabilities for both large losses and large gains. If $X$ is the possible random loss in future, then the distortion risk measure defined above represents the decision maker's attitude: risk aversion for large losses and risk-seeking for small losses or gains; %. The decision makers employing those distortion risk measures to evaluate the risk are risk-averse for large losses and risk-seeking for large gains;
see \cite{Y87} and \cite{TK92}. The distortion risk measures with general distortion functions are also  popular in insurance pricing, performance evaluation and other applications; see e.g., \cite{Wang96}, \cite{Wang00}, \cite{CM09} and \cite{FS16}. Note that the risk sharing problem for $\lambda \mathbb E(X)+(1-\lambda) R_{I_i}(X)$ is equivalent to the risk sharing problem for $R_{I_i}$ in terms of the optimal risk allocations. This motivates us to consider the risk sharing problem for $R_{I_i}$.

For monetary risk measures\footnote{We say a mapping $\rho:\X \to \R$ is a monetary risk measure if it satisfies cash invariance, i.e., 
$\rho(X+c)=\rho(X)+c$ for $X\in \X$ and $c\in \R$, and monotonicity, i.e., $\rho(X)\le\rho(Y)$ for $X\le Y$; see e.g., Chapter 4 of \cite{FS16}.},
 optimality with respect to the sum (referred to as the sum optimality) is equivalent to Pareto optimality (e.g., Proposition 1 of \cite{ELW18}). Clearly, $R_I$ and $\rho_{g_{\lambda,i}}$ are  monetary risk measures. Further economic interpretations on the inf-convolution can be found in e.g., Chapter 10 of  \cite{R13}.

We next give the definition of comonotonicity and counter-monotonicity for $n$ random variables to interpret the optimal allocations. We say $(X_1,\dots,X_n)$ are \emph{comonotonic} if there exist increasing functions $f_i,~i\in [n]$ such that $X_i=f_i(X_1+\dots+X_n)$ for all $i\in [n]$; we say  $(X_1,\dots,X_n)$ are \emph{counter-monotonic} if $(-X_i, X_j)$ are comonotonic for any $i\neq j$.   We refer to \cite{DDGKV02} and \cite{JKLW22} for an overview on comonotonic risk sharing rules, and \cite{LLW23} for the characterization on counter-monotonicity.

Our result for Problem \eqref{Eq:inf-convolution} is displayed as follows. In what follows, $x_+:=x\vee 0$ for $x\in\R$. %in the following theorem.

%For $X \in \X$ bounded from above and  $\bsyb{\beta}\in [0, 1)^n $, we define
% In \cite{WZ20}, $A^c$ is called a $\beta$-tail event of $X$.
%Next, we give the dependence structure of $(X_1,\dots,X_n,X)$ such that 

\begin{theorem}[Risk sharing for the average quantile] 
\label{thm: SumOptimalReducedNewRVaR}
	For  $X \in \X_1$ and  $\bsyb{\beta}\in (0,1)^n$ satisfying $\beta\in (0,1)$, we have
	\begin{equation}
	\label{eqn: ReducedNewRVaR}
	\dsquare_{i=1}^n R_{I_i%[0, \beta_i] \cup [1-\beta+\beta_i, 1]
    }(X) = R_{[0,  \beta]}(X).
	\end{equation} 
	Moreover, if $X$ is bounded from above, an optimal allocation is given by 
	\begin{align}%\label{eqn: YLOptimalAllocation}
X_i =  (X-t)  \id_{A_i}   + \frac{X}n \id_{A^c} + \frac{t}{n-1}\id_{A\setminus A_i},~~i\in [n], \label{eq:allocation}
\end{align}
where $t\geq  (q_1^-(X))_+$, $A = \{U_X \leq \beta \}$ and $(A_1,\dots,A_n)$ is a partition of $A$ satisfying $\p(A_i)=\beta_i$ for all  $i\in [n]$.
\end{theorem}

\begin{proof}
	By Corollary \ref{coro: reducedRVaR}, we have
	\begin{align}\label{eq:inf-conv2}
	\dsquare_{i=1}^n R_{I_i%[0, \beta_i] \cup [1-\beta+\beta_i, 1]
    }(X) \ge  R_{[0,  \beta]}(X).
	\end{align} 
	Hence, it suffices to show the inverse inequality of \eqref{eq:inf-conv2}.
	Suppose that $X$ is bounded from above and $(X_1,\dots,X_n)$ is given by \eqref{eq:allocation}.
	Note that 
	\begin{align*}
	\sum_{i=1}^n X_i  & =   \sum_{i=1}^n  (X-t)  \id_{A_i}   +   \sum_{i=1}^n  \frac{X}n \id_{A^c} +    \sum_{i=1}^n \frac{t}{n-1}\id_{A\setminus A_i}
	\\&= (X-t) \id_{A} + X \id_{A^c} + t \id_{A}  =X,
	\end{align*}
	and thus $(X_1,\dots,X_n)\in \mathbb A_n(X)$.  
	We claim that, for each $i\in [n]$, 
	\begin{align}\label{eq:order} 
	X_i(\omega_1)\le X_i(\omega_2) \le X_i(\omega_3) 
	\mbox{~~~for  $\omega_1\in A_i$, $\omega_2\in A^c$ and $\omega_3\in A\setminus A_i$ almost surely.}
	\end{align} 
	 Note that  $X_i(\omega_2)   \le    X_i(\omega_3)$ holds trivially.
     We next show $X_i(\omega_1)   \le    X_i(\omega_2)$. By our construction, we have
	$X_i(\omega_1)  \le q_\beta^+(X)-t$
	and $X_i(\omega_2) \in \left[ q_{\beta}^+(X)/n ,q_{1}^-(X)/n \right] $. 
		If $q_{\beta}^+ (X)\le 0$, then 
		$$X_i(\omega_1)  \le q_\beta^+(X)-t \le q_\beta^+(X) 
		\le \frac{q_\beta^+(X)}n \le X_i(\omega_2).$$  
		If $q_{\beta}^+ (X)> 0$, then 
		$$X_i(\omega_1)  \le q_\beta^+(X)-t\le 0
		\le \frac{q_\beta^+(X)}n \le X_i(\omega_2).$$  
	 Using \eqref{eq:order}, we have 
	\begin{align*} R_{I_i%[0, \beta_i] \cup [1-\beta+\beta_i, 1]
    }(X_i)&  = \frac{1}{\beta} \E\left[ (X-t) \id_{A_i} +  \frac{t}{n-1}\id_{A\setminus A_i}\right] 
	\\&=  \frac{1}{\beta} \E[ X \id_{A_i}] - \frac{\beta_i}{\beta} t +  \frac{ \beta-\beta_i }{(n-1)\beta}t 
	\\&=  \frac{1}{\beta} \E[ X \id_{A_i}] +   \frac{\beta-n\beta_i}{(n-1)\beta}t.
	\end{align*}
	It follows that 
	$$ \sum_{i=1}^n R_{I_i%[0, \beta_i] \cup [1 - \beta + \beta_i, 1]
    }(X_i) = \sum_{i=1}^n \frac{1}{\beta} \E[ X \id_{A_i}]
	+  \sum_{i=1}^n \frac{\beta-n\beta_i}{(n-1)\beta}t
	=\frac{1}{\beta} \E[ X \id_A]= R_{[0,  \beta]}(X).$$
	Hence,  we obtain the inverse inequality of \eqref{eq:inf-conv2}, implying that \eqref{eqn: ReducedNewRVaR} and \eqref{eq:allocation} hold for $X$ being bounded from above.
	
	Next, we consider the case that $X$ is not bounded from above. For $m\geq 1$, 
	let $X^{(m)}=X \wedge m$ and $Z^{(m)}=X-X^{(m)}=(X-m)\id_{\{X>m\}}$. Using the above result, we know that for each $m \geq 1$, there exists $(X_1^{(m)},\dots,X_n^{(m)})\in \mathbb A_n(X^{(m)})$ such that 
	$$ \sum_{i=1}^n R_{I_i%[0, \beta_i] \cup [1-\beta+\beta_i, 1]
    } \left( X_i^{(m)}\right) =   R_{[0,  \beta]}\left(X^{(m)}\right) \le R_{[0,  \beta]}\left(X \).
    $$ 
	Let $Y_1^{(m)}=X_1^{(m)}+Z^{(m)}$. Then
	we have $\left(Y_1^{(m)}, X_2^{(m)},\dots,X_n^{(m)}\right)\in \mathbb A_n(X)$.
    In light of Theorem 1 of \cite{ELW18}, we have
	\begin{align*}
	&\quad R_{I_1%[0, \beta_1) \cup (1-\beta+\beta_1, 1]
    }\left(X_1^{(m)}+Z^{(m)}\right)
	\\ & = \frac{\beta_1 }{\beta}R_{[0, \beta_1]} \left(X_1^{(m)}+Z^{(m)}\right) + 
	\frac{\beta - \beta_1 }{\beta}
	R_{ [1-\beta+\beta_1, 1]}\left(X_1^{(m)}+Z^{(m)}\right)
	\\ &\le \frac{\beta_1 }{\beta}\left(R_{[0, \beta_1]} \left(X_1^{(m)}\right) +\ES_{\beta_1}\left(Z^{(m)}\right) \right) + 
	\frac{\beta - \beta_1 }{\beta}\left( 
	R_{ [1-\beta+\beta_1, 1]}\left(X_1^{(m)} \right) +\ES_{\beta-\beta_1}\left(Z^{(m)}\right) \right)
	\\ & =  R_{I_1%[0, \beta_1) \cup (1-\beta+\beta_1, 1]
    }\left(X_1^{(m)}\right) + \frac{\beta_1 }{\beta} \ES_{\beta_1} \left(Z^{(m)}\right)  + \frac{\beta - \beta_1 }{\beta} \ES_{\beta-\beta_1}\left(Z^{(m)}\right).
	\end{align*} 
	Note that 
	$$ \frac{\beta_1 }{\beta} \ES_{\beta_1} \left(Z^{(m)}\right)  + \frac{\beta - \beta_1 }{\beta} \ES_{\beta-\beta_1}\left(Z^{(m)} \right) \le \frac{2}{\beta}\E\left[Z^{(m)}\right] \to 0~\text{as}~ m\to \infty. $$
	Therefore, for any $\epsilon>0$, there exists $m_0>1$ such that 
	$$R_{I_1%[0, \beta_1) \cup (1-\beta+\beta_1, 1]
    }\left(Y_1^{(m)}\right)   \le R_{I_1%[0, \beta_1) \cup (1-\beta+\beta_1, 1]
    }\left(X_1^{(m)} \right)+  \epsilon$$
	for all $m>m_0$, which  implies
	$$R_{I_1%[0, \beta_1] \cup [1-\beta+\beta_1, 1]
    }\left(Y_1^{(m)}\right)+ \sum_{i=2}^n R_{I_i%[0, \beta_i] \cup [1-\beta+\beta_i, 1]
    }(X_i^{(m)}) \le  \sum_{i=1}^n R_{I_i%[0, \beta_i] \cup [1-\beta+\beta_i, 1]
    }\left(X_i^{(m)}\right) + \epsilon \le R_{[0,  \beta]}\left(X \right) +  \epsilon.$$ 
	 By the arbitrariness of $\epsilon$, we obtain $\dsquare_{i=1}^n R_{I_i%[0, \beta_i] \cup [1-\beta+\beta_i, 1]
    }(X) \leq  R_{[0,  \beta]}(X)$, which together with \eqref{eq:inf-conv2} implies \eqref{eqn: ReducedNewRVaR}. The optimal allocation given in \eqref{eq:allocation} has been checked in the above proof.  This completes the proof.
\end{proof}

Note that \eqref{eq:allocation} is also an optimal allocation if we replace $t$ in \eqref{eq:allocation} by a random variable $Y$ satisfying  $Y\geq  (q_1^-(X))_+$. There are also other types of optimal allocation; see  Proposition \ref{cor:extremal} below. The optimal allocation in \eqref{eq:allocation} shows that the risk is equally allocated over $A^c$ and is counter-monotonic over $A$, i.e.,  
$(X-t) \id_{A_i}+\frac{t}{n-1}\id_{A\setminus A_i},~~i\in [n]$ with $(A_1,\dots,A_n)$ being a partition of $A$ satisfying $\p(A_i)=\beta_i$ for all  $i\in [n]$ are counter-monotonic restricted to $A$.

In the literature of risk sharing, the optimal allocation is comonotonic if the risk functionals are law-invariant convex risk measures (see e.g., \cite{JST08} and \cite{FS08}), and is counter-monotonic if the risk functionals are quantile-based risk measures (see e.g., \cite{ELW18}, \cite{ELMW20} and \cite{LMWW22}).   The optimal allocation in \eqref{eq:allocation} exhibits a combination of two types of risk sharing in the literature: comonotonic risk sharing for large losses and counter-monotonic risk sharing for small losses or large gains, which may be due to the agents' risk attitudes:  risk-aversion  for large losses and risk-seeking  for small losses or large gains.

In Theorem \ref{thm: SumOptimalReducedNewRVaR}, we only give the optimal allocation in \eqref{eq:allocation} if $X$ is bounded from above. The existence of the optimal allocation is unknown from Theorem \ref{thm: SumOptimalReducedNewRVaR} if $X$ is unbounded from above. To answer this question requires the analysis on the dependence structure of the optimal allocations, which is not trivial.

Next, we discuss the dependence structure of $(X_1,\dots,X_n,X)$ with $(X_1,\dots,X_n)\in \mathbb A_n(X)$ such that for  $\bsyb{\beta}\in (0,1)^n$ satisfying $\beta = \sum_{i=1}^n \beta_i \in (0,1)$, \eqref{eqn: ReducedNewRVaR} holds.
%\begin{align}\label{Eq:minimal}\sum_{i=1}^n R_{[0, \beta_i] \cup [1-\beta+\beta_i, 1]}(X_i) =R_{[0, \beta]}(X).
%\end{align}
\begin{proposition}
\label{cor:extremal}
    For $X\in\mathcal X_1$ and
$(X_1,\ldots,X_n)\in\mathbb{A}_n(X)$, we have
$$
\sum_{i=1}^n R_{I_i}(X_i)=R_{[0,\beta]}(X)
$$
if and only if there exist $U_{X_1},\ldots,U_{X_n}$ and $U_X$ such that
\begin{equation}
        \label{Eq:dependencestructure}
        \left\{U_X \in \[0,\beta\] \right\}= \left\{ U_{X_i} \in I_i%[0, \beta_i] \cup [1-\beta+\beta_i, 1] 
        \right\}, \quad i \in \[n\].
    \end{equation}  
    Moreover, \eqref{Eq:dependencestructure} implies that one  of the following two statements holds:
    \begin{enumerate}[(i)]
        \item \label{enu: EDS RVaR inequality case 1}
        $\bigcup_{i =1}^n \{U_{X_i} \in [0,\beta_i] \} = \{ U_X \in [0,\beta]\}$;
        \item \label{enu: EDS RVaR inequality case 2}
        The set $\{ U_X \in [0,\beta] \} \setminus \bigcup_{i = 1}^n \{U_{X_i} \in [0,\beta_i] \}$ has a positive probability, which is denoted by $\theta$. Meanwhile, random variables $X_1, \cdots, X_n$ and $X$ are all constants on the set $\{U_X \in [\beta-\theta, 1]\}$, which actually equals the set $\{ U_{X_i} \in [\beta_i, 1-\beta+\beta_i+\theta]\}$ for all $i \in [n]$. 
       % Furthermore, the set $\bigcap_{i = 1}^n \{ U_{X_i} \in [0,\beta_i]\} = \emptyset$.
    \end{enumerate}
\end{proposition}
\begin{proof}
    Obviously, the equation \eqref{Eq:dependencestructure} is a sufficient condition for %\eqref{eqn: ReducedNewRVaR}. 
    $
    \sum_{i=1}^n R_{I_i}(X_i)=R_{[0,\beta]}(X).
    $ 
    We next show the only if part. Let $\beta_0=1-\beta$. Then we have $\sum_{i = 0}^n \beta_i = 1$. The equality
$$
\sum_{i=1}^n R_{I_i}(X_i)=R_{[0,\beta]}(X)
$$
is equivalent to
%Note that  \eqref{eqn: ReducedNewRVaR} is equivalent to 
$\sum_{i=1}^n R_{[\beta_i, 1-\beta+\beta_i]}(X_i) =R_{[\beta,1]}(X),$
        which is further equivalent to 
        $$\sum_{i=1}^n R_{[1-\beta_i-\beta_0, 1-\beta_i]}(-X_i) =R_{[0,\beta_0]}(-X).$$
        Let $Y_i=-X_i,~i\in [n]$ and $Y=-X$.
        For $i \in [n]$,  define random variables $$T_i := Y_i \mathds{1}_{\{U_{Y_i} \leq 1 - \beta_i\}} + m \mathds{1}_{\{U_{Y_i} > 1 - \beta_i\}},$$ where $m\in\mathbb R$ satisfying $m<\bigwedge_{i=1}^n q_{1-\beta_i-\beta_0}^-(Y_i)$.
         By the construction of $T_i$, one could see that for any $s \in \Rb$,
    $$
    \begin{aligned}
        \Pb\left(\sum_{i=1}^n T_i >s\right) \geq & \Pb\left(\left\{\sum_{i=1}^n Y_i > s\right\} \cap (\cup_{i=1}^n \{U_{Y_i} > 1 - \beta_i\})^c \right) \\
        \geq & \Pb\left(\sum_{i=1}^n Y_i > s\right) - \Pb\left(\cup_{i=1}^n \{U_{Y_i} > 1 - \beta_i\} \right) \\
        \geq & \Pb\left(\sum_{i=1}^n Y_i > s\right) - 1 + \beta_0.    
    \end{aligned}
    $$
    Thus, we have
    $$
        q_u^{-}\left(\sum_{i=1}^n T_i\right) \geqslant q_{u-1+\beta_0}^{-}\left(\sum_{i=1}^n Y_i\right), \quad \forall u\in(1-\beta_0,1].
    $$
    Then, combined with the subadditivity of $\ES$, we have
    \begin{equation*}
        % \label{eqn: Sharp Subadditivity of RVaR}
    \begin{aligned}
        \sum_{i=1}^n R_{\[1-\beta_i-\beta_0, 1-\beta_i\]}(Y_i)=\sum_{i=1}^n \mathrm{ES}_{\beta_0}(T_i) & \geqslant \mathrm{ES}_{\beta_0}\left(\sum_{i=1}^n T_i\right) \\
        & =\frac{1}{\beta_0} \int_{1-\beta_0}^1 q_u^{-}\left(\sum_{i=1}^n T_i\right) \mathrm{d} u \\
        & \geqslant \frac{1}{\beta_0} \int_{1-\beta_0}^1 q_{u-1+\beta_0}^{-}\left(Y\right) \mathrm{d} u.
        \end{aligned}
    \end{equation*}   
    Direct computation shows 
    \begin{equation*}
    %\label{eqn: Sharp Subadditivity of RVaR}
    \begin{aligned}
        \frac{1}{\beta_0} \int_{1-\beta_0}^1 q_{u-1+\beta_0}^{-}\left(Y\right) \mathrm{d} u
        & =\frac{1}{\beta_0} \int_0^{\beta_0} q_u^{-}\left(Y\right) \mathrm{d} u \\
        & =\mathrm{LES}_{\beta_0}\left(Y\right) \\
        % & \geq \mathrm{LES}_t\left(\sum_{i=1}^n Y_i\right) \\
        & = \sum_{i=1}^n R_{\[1-\beta_i-\beta_0, 1-\beta_i\]}(Y_i),
    \end{aligned}
    \end{equation*}
    which implies that $\sum_{i=1}^n \mathrm{ES}_{\beta_0}(T_i) = \mathrm{ES}_{\beta_0}(\sum_{i=1}^n T_i)=\mathrm{LES}_{\beta_0}\left(\sum_{i=1}^n Y_i\right)$. By Theorem 5 of \cite{WZ21} and the fact  $\sum_{i=1}^n \mathrm{ES}_{\beta_0}(T_i) = \mathrm{ES}_{\beta_0}(\sum_{i=1}^n T_i)$, there exist $U_{T_i} $,  $i \in [n]$ and $U_{\sum_{i=1}^n T_i}$ such that $\{U_{T_i} \in [1 -\beta_0, 1]\}= \{U_{\sum_{i=1}^n T_i}\in [1 -\beta_0, 1]\}$ for all $i \in [n]$. 
    Hence, we have  
     \begin{align*}\mathrm{ES}_{\beta_0}\left(\sum_{i=1}^n T_i\right)
    &=\frac{1}{\beta_0}\mathbb E\left(\sum_{i=1}^n T_i
    \id_{\{U_{T_i} \in \[1 -\beta_0, 1\]\}}\right)\\
    &=\frac{1}{\beta_0}\mathbb E\left(\sum_{i=1}^n Y_i
    \id_{\{U_{T_i} \in \[1 -\beta_0, 1\]\}}\right)\\
    &=\frac{1}{\beta_0}\mathbb 
    E(Y\id_{\{U_{T_i} \in \[1 -\beta_0, 1\]\}}) 
    =\mathrm{LES}_{\beta_0}(Y).
    \end{align*}
     This implies that there exists  $U_Y$ such that $\{U_Y\in [0, \beta_0] \} = \{ U_{T_i}\in [1-\beta_0, 1] \}=\{U_{Y_i}\in [1-\beta_i-\beta_0,1-\beta_i]\}$ for all $i\in [n]$. Note that $U_{X}=1-U_Y$ and $U_{X_i}=1-U_{Y_i}$. Hence, we have  $\{U_X\in [1-\beta_0,1] \} =\{U_{X_i}\in [\beta_i, \beta_i+\beta_0]\}$ for all $i\in [n]$, which is equivalent to \eqref{Eq:dependencestructure}. We establish the first claim.
        
    Note that \eqref{Eq:dependencestructure} implies $\bigcup_{i=1}^n \{U_{X_i} \in [0,\beta_i]\} \subseteq \{ U_X\in [0,\beta]\}$ and $\bigcup_{i=1}^n  \{U_{X_i} \in [\beta_i+\beta_0,1]\} \subseteq \{ U_X\in [0, \beta]\}$.
    Thus, it would be either case (i) or (ii). 
    For the latter case, i.e., $\mathbb P(\{ U_X \in [0,\beta] \} \setminus \bigcup_{i = 1}^n \{U_{X_i} \in [0,\beta_i] \})>0$, \eqref{Eq:dependencestructure} implies 
    \begin{equation}
    \label{eqn: EDS coro Essential Extra Case}
       \{ U_X \in [0,\beta] \} \setminus \bigcup_{i = 1}^n \{U_{X_i} \in [0,\beta_i] \} =\bigcap_{i=1}^n \{U_{X_i} \in [\beta_i+ \beta_0, 1]\}.
    \end{equation} 
    Then, for any $\omega \in \{ U_X\in [0,\beta]\} \setminus \bigcup_{i=1}^n \{U_{X_i} \in [0,\beta_i]\} $ and any $\omega^\prime \in \{U_X\in [\beta,1] \}$, it follows from  \eqref{Eq:dependencestructure} and \eqref{eqn: EDS coro Essential Extra Case} that $X_i(\omega) \geq  X_i(\omega^\prime)$ holds for all $i \in [n]$. On the other hand, using \eqref{eqn: EDS coro Essential Extra Case}, one has that $\sum_{i=1}^n X_i(\omega) \leq  \sum_{i=1}^n X_i(\omega^\prime)$. 
    Thus all the above inequalities hold as equalities. By the arbitrariness of $\omega$ and $\omega^\prime$, we conclude that random variables $X_1, \cdots, X_n$ and $X$ are all constants on $\{U_X\in [\beta, 1]\} \cup (\{ U_X\in [0,\beta]\} \setminus \bigcup_{i=1}^n \{U_{X_i} \in [0,\beta_i]\})$. Hence, there exist $U_{X_1},\dots, U_{X_n}$ and $U_X$ such that case (ii) holds.
  % \trd{The case  $\sum_{i=0}^n \beta_i <1$ can be proved analogously as the case $\sum_{i=0}^n \beta_i=1$ and the details are omitted. }
   % To prove the general case that $\sum_{i=0}^n \beta_i <1$, one restricts $\Xf$ onto the set $\{U_X \geq 1 - \sum_{i =1}^n \beta_i + \beta_0\}$ and then the equation \eqref{eqn: Simplified & Reduced Convolution Bound} holds for the obtained random vector with the parameter vector $\frac{\boldsymbol{\beta} - \sum_{i=0}^n \beta_i}{\sum_{i=0}^n \beta_i}$. 
\end{proof}

In the following Theorem \ref{Prop:nonexistence}, we show that the optimal allocation for \eqref{eqn: ReducedNewRVaR} does not exist if $X$ is not bounded from above, while we find a sequence of allocations $\{(X_1^{(m)},\dots,X_n^{(m)})\}_{m \geq 1}$ such that the sum risk exposure converges to the lower bound $R_{[0,\beta]}(X)$. 
% $$\lim_{m\to\infty}\sum_{i=1}^n R_{I_i%[0, \beta_i] \cup[1-\beta+\beta_i,1]
% }\left(X_i^{(m)}\right)=R_{[0,\beta]}(X).$$ 
For $m\geq 1$, let \begin{align}\label{eqn:allocationlimit}
X_i^{(m)} =  \(X-m\)  \id_{A_i}   + \frac{X}n \id_{A^c} + \frac{m}{n-1}\id_{A\setminus A_i},~~i\in [n], 
\end{align}
where  $A = \{U_X \leq \beta \}$ and $(A_1,\dots,A_n)$ being a partition of $A$ satisfying $\p(A_i)=\beta_i$ for all  $i\in [n]$.
\begin{theorem}[Non-existence of the optimal allocation]
\label{Prop:nonexistence}
    For $X \in \Xc_1$ and $\bsyb{\beta}\in (0,1)^n$  satisfying $\beta\in (0, 1)$, if $X$ is not bounded from above, then 
    \begin{equation*}
        %\label{eqn: No Lower Convolution Sharpner for Total Risk not Bounded from above}
        \sum_{i=1}^n R_{I_i%\left[0, \beta_i\right] \cup\left[1 - (\beta - \beta_i), 1\right]
        }\left(X_i\right)> R_{[0, \beta]}(X)
    \end{equation*}
    holds for all  $\left(X_1, \ldots, X_n\right) \in \mathbb{A}_n(X)$. 
    Moreover, the risk allocation $\left(X_1^{(m)},\dots,X_n^{(m)}\right) \in \A_n(X)$, with $m\geq 1$, defined in \eqref{eqn:allocationlimit} satisfies 
    $$\sum_{i=1}^n R_{I_i}(X_i^{(m)})=R_{[0,\beta]}(X)+\frac{1}{\beta}\E((X-a_m)\id_{\{X>a_m\}})$$
   for $m\geq  (q_{\beta}^-(X)\vee \bigvee_{i=1}^n q^-_{1-\beta+\beta_i}(X))_+$ with $a_m=\frac{nm}{n-1}$, and 
  $$\lim_{m\to\infty}\sum_{i=1}^n R_{I_i}(X_i^{(m)})=R_{[0,\beta]}\(X\).$$ 
\end{theorem}

\begin{proof}
    Suppose that the total risk $X$ is not bounded from above % and  that %the equality \eqref{eqn: ReducedNewRVaR} holds for a risk allocation $\left(X_1, \ldots, X_n\right) \in \mathbb{A}_n(X)$.
    and that there exists an allocation $(X_1,\ldots,X_n)\in\mathbb{A}_n(X)$
such that
$
\sum_{i=1}^n R_{I_i}(X_i)=R_{[0,\beta]}(X).
$  Then, by Proposition \ref{cor:extremal}, with $\beta_0=1-\beta$,
    there exist $U_{X_i}$, $i \in [n]$ and $U_X$ such that 
    $$\{U_X\in [0, \beta]\} = \{U_{X_i} \in I_i%[0, \beta_i] \cup [\beta_i +1- \beta, 1] 
    \}, \quad i \in [n],$$
    which is equivalent to
    $$
    \{ U_X\in [\beta, 1]\} = \{ U_{X_i} \in [\beta_i, 1 -\beta +\beta_i]\}, \quad i \in [n].
    $$
    Then we have
    $$
    \infty = \esssup X\id_{\{ U_X\in [\beta, 1]\}}=\esssup\sum_{i=1}^n X_i\id_{\{ U_{X_i} \in [\beta_i, 1 -\beta+\beta_i]\}}
    \leq\sum_{i=1}^n q_{1-\beta+\beta_i}^-(X_i)<\infty,
    $$
    which leads to a contradiction. %Hence, there is no solution for \eqref{eqn: ReducedNewRVaR} if $X$ is unbounded from above. 
    Hence, no allocation can attain the value $R_{[0,\beta]}(X)$ if $X$ is
unbounded from above.
    
    Next, we focus on the second statement.  For $m\geq (q_{\beta}^-(X)\vee \bigvee_{i=1}^n q^-_{1-\beta+\beta_i}(X))_+$, direct computation shows
    \begin{align*}
   R_{I_i}(X_i^{(m)}) =  \frac{1}{\beta}\E\left((X-m) \id_{A_i}  + \left(\frac{X}{n}\vee \frac{m}{n-1}\right)\id_{\{U_X>1-\beta+\beta_i\}}\right),~i\in [n].
    \end{align*}
    Hence, we have 
    \begin{align*}
   \sum_{i=1}^nR_{I_i}(X_i^{(m)}) &=  \frac{1}{\beta}\E\left((X-m) \id_{A}  +\left(\frac{X}{n}\vee \frac{m}{n-1}\right)\sum_{i=1}^n\id_{\{U_X>1-\beta+\beta_i\}}\right)\\
   &=R_{[0,\beta]}(X)-m+\frac{1}{\beta}\sum_{i=1}^n \left(\E\left(\frac{X}{n}\id_{\{X>a_m\}}\right)+\frac{m}{n-1}(\beta-\beta_i-\p(X>a_m))\right)\\
   &=R_{[0,\beta]}(X)+\frac{1}{\beta}\E\left(X\id_{\{X>a_m\}}\right)-\frac{a_m}{\beta}\p(X>a_m),
    \end{align*}
    where $a_m=\frac{nm}{n-1}$.
    Note that $a_m\p(X>a_m)\leq \E(X\id_{\{X>a_m\})}\to 0$ as $m\to \infty$. Hence, we have 
    $$\lim_{m\to\infty}\sum_{i=1}^n R_{I_i}(X_i^{(m)})=R_{[0,\beta]}(X).$$ 
  This completes the proofs.
\end{proof}

%Although the optimal allocation does not exist when $X$ is not bounded from above,  we could find a sequence of allocations $\{(X_1^{(m)},\dots,X_n^{(m)}),m\geq 1\}$ such that
%$$\lim_{m\to\infty}\sum_{i=1}^n R_{[0, \beta_i] \cup[1-\beta+\beta_i,1]}(X_i^{(m)})=R_{[0,\beta]}(X).$$ 
%An example of such allocation is given by
%\begin{align*}%\label{eqn: YLOptimalAllocation}
%X_i^{(m)} =  (X-m)  \id_{A_i}   + \frac{X}n \id_{A^c} + \frac{m}{n-1}\id_{A\setminus A_i},~~i\in [n], 
%\end{align*}
%where  $A = \{U_X \leq \beta \}$ and $(A_1,\dots,A_n)$ is a partition of $A$ satisfying $\p(A_i)=\beta_i$ for all  $i\in [n]$.

In light of Theorems \ref{thm: SumOptimalReducedNewRVaR} and  \ref{Prop:nonexistence}, we obtain the following results on risk sharing for distortion risk measures with  special inverse S-shaped distortion functions.
\begin{proposition}\label{prop:distortion}
    For  $X \in \X_1$ and  $\bsyb{\beta}\in (0,1)^n$ satisfying $\beta\in (0,1)$ and $\lambda\in (0,1)$, we have
	\begin{equation*}
	\dsquare_{i=1}^n\rho_{g_{\lambda,i}} (X) = \lambda \E(X)+(1-\lambda) R_{[0,  \beta]}(X).
	\end{equation*} 
	Moreover, if $X$ is bounded from above, an optimal allocation is given by \eqref{eq:allocation}; if $X$ is not bounded from above, the optimal allocation does not exist and the risk allocation $\left(X_1^{(m)},\dots,X_n^{(m)}\right) \in \A_n(X)$, with $m\geq 1$, defined in \eqref{eqn:allocationlimit} satisfies 
  $$\lim_{m\to\infty}\sum_{i=1}^n \rho_{g_{\lambda,i}}\(X_i^{(m)}\)=\dsquare_{i=1}^n\rho_{g_{\lambda,i}} (X).$$ 
\end{proposition}

Proposition \ref{prop:distortion} shows that the Pareto-optimal risk allocation for $\rho_{g_{\lambda,i}}$ is the combination of comonotonic risk sharing for large losses and counter-monotonic risk sharing for small losses or large gains, consistent with the risk preference represented by $\rho_{g_{\lambda,i}}$: risk aversion for large losses and risk-seeking for small losses or large gains. Finding the optimal risk sharing for distortion risk measures with inverse S-shaped distortion functions is a very challenging  problem due to the non-convexity of the corresponding distortion risk measures.  Although Proposition \ref{prop:distortion} only solves a special case, to the best of our knowledge, it is the first result offering the optimal allocations for the non-constrained risk sharing with this class of distortion risk measures with its distortion functions exaggerating the probabilities of large losses and large gains simultaneously. This can be seen from the novelty of the optimal allocations and the condition for the existence of the optimal allocations. 
Proposition \ref{prop:distortion} suggests that the optimal allocation exists if and only if $X$ is bounded from above. %We conjecture that this conclusion holds true for the distortion risk measures with general inverse S-shaped distortion functions.
However, the risk sharing problem for distortion risk measures with general inverse S-shaped distortion functions is still unknown and will be studied in future.

Note that our result in Theorem \ref{thm: SumOptimalReducedNewRVaR} solves the dual problem of the one in Theorem 2 of \cite{ELW18}.
We observe \begin{align*}\dsquare_{i=1}^n R_{I_i%[0, \beta_i] \cup [1-\beta+\beta_i, 1]
}(X )&=\inf\left\{\sum_{i=1}^n R_{I_i%[0, \beta_i] \cup [1-\beta+\beta_i, 1]
}(X_i): (X_1,\dots, X_n)\in \mathbb A_n(X)\right\}\\
&=\inf\left\{\sum_{i=1}^n \left(\frac{\mathbb E(X_i)}{\beta}-\frac{1-\beta}{\beta}R_{[\beta_i, 1-\beta+\beta_i]}(X_i)\right): (X_1,\dots, X_n)\in \mathbb A_n(X)\right\}\\
&=\frac{\mathbb E(X)}{\beta}-\frac{1-\beta}{\beta}\sup\left\{\sum_{i=1}^n R_{[\beta_i, 1-\beta+\beta_i]}(X_i): (X_1,\dots, X_n)\in \mathbb A_n(X)\right\}.
\end{align*}
Hence, the optimal allocation for $\dsquare_{i=1}^n R_{I_i%[0, \beta_i] \cup [1-\beta+\beta_i, 1]
}(X)$ is the worst allocation for $\sum_{i=1}^n R_{[\beta_i, 1-\beta+\beta_i]}(X_i)$ with $(X_1,\dots, X_n)\in \mathbb A_n(X)$. Using Theorems \ref{thm: SumOptimalReducedNewRVaR} and  \ref{Prop:nonexistence}, we have  the following conclusion.
\begin{proposition} For  $X \in \X_1$ and  $\bsyb{\beta}\in (0,1)^n$ satisfying $\beta\in (0,1)$, we have
$$\sup\left\{\sum_{i=1}^n R_{[\beta_i, 1-\beta+\beta_i]}(X_i): (X_1,\dots, X_n)\in \mathbb A_n(X)\right\}=\frac{\mathbb E(X)-\beta R_{[0,  \beta]}(X)}{1-\beta}=R_{[\beta,1]}(X).$$
Moreover, if $X$ is bounded from above, an allocation given by \eqref{eq:allocation} achieves the supremum; if $X$ is not bounded from above, no allocation can achieve the supremum.
\end{proposition}

\section{Conclusion}\label{sec:conc}

In this paper, we obtain a new $\RVaR$ inequality, differing from the one in \cite{ELW18}. Applying this new inequality, we obtain the upper and lower bounds for robust $\RVaR$ by assuming fixed marginal distributions and unknown dependence structure, which is sharp if the marginal distributions have increasing densities on their upper-tail parts for the upper bounds and if the marginal distributions have decreasing densities on their upper-tail part for the lower bounds. Those bounds complement the results and fill in some  gaps of \cite{BLLW24}. Moreover, we obtain the sharp upper bounds for the difference between two $\RVaR$s and the difference between two quantiles, extending the sharp bounds on $\RVaR$ and quantiles to the corresponding variability risk measures. The application of those extended convolution bounds in portfolio optimization and optimal insurance and reinsurance is left for future investigation. 
Finally, applying the new inequality, we obtain the Pareto-optimal risk allocation for some non-convex averaged quantiles, which corresponds to the Pareto-optimal risk allocation for distortion risk measures with special inverse S-shaped distortion functions exaggerating both the probability of large losses and the probability of large gains. By analyzing the dependence structure of the optimal risk allocation, we show that the optimal allocation does not exist if the risk is unbounded from above. However, we offer a sequence of allocations whose aggregate risk exposure converges to the inf-convolution.  The Pareto-optimal risk sharing for distortion risk measures with general inverse S-shaped distortion functions is still an open problem due to the non-convex nature of the distortion risk measures.

\subsubsection*{Acknowledgement} 
The authors are grateful to Ruodu Wang and members of the research group on financial mathematics at The Chinese University of Hong Kong (Shenzhen) for their useful feedback and conversations.
%The authors would like to thank the anonymous referee for his/her insightful suggestions which have helped us improve the paper.
Y. Liu acknowledges financial support from the National Natural Science Foundation of China (Grant No. 12401624), The Chinese University of Hong Kong (Shenzhen) University Development Fund (Grant No. UDF01003336), Guangdong Science and Technology Program (Grant No. 2024QN11X076) and Shenzhen Science and Technology Program (Grant No. RCBS20231211090814028, JCYJ20250604141203005, 2025TC0010) and is partly supported by the Guangdong Provincial Key Laboratory of Mathematical Foundations for Artificial Intelligence (Grant No. 2023B1212010001). 
% =============================================================================
% BIB
% =============================================================================
 	%\nocite{*}

\end{document}